\documentclass[a4paper,USenglish,cleveref,thm-restate,numberwithinsect]{lipics-v2021}
\usepackage{tikz}
\usepackage{mathtools}
\usepackage[noend,ruled,linesnumbered]{algorithm2e}
\usepackage[space]{cite} 
\usepackage{etoolbox}

\makeatletter
\patchcmd{\@algocf@start}
  {-1.5em}
  {0pt}
  {}{}
\makeatother

\newif\ifArXiv
\ArXivtrue

\usepackage[bb=dsserif]{mathalpha}

\crefalias{AlgoLine}{line}
\crefname{line}{Line}{Lines}

\ifArXiv%
\hypersetup{colorlinks=true,allcolors=blue}
\fi%

\newtheorem{fact}[theorem]{Fact}
\newtheorem{problem}[theorem]{Problem}

\newtheorem{construction}[theorem]{Construction}
\newcommand{\set}[1]{\left\{#1\right\}}

\newcommand\tuple[1]{\langle {#1}\rangle}

\newcommand{\EDA}{\ensuremath{\mathtt{ED}_a}}
\newcommand{\ED}{\mathtt{ED}}
\newcommand{\HD}{\mathtt{HD}}

\newcommand{\OV}{\mathsf{OV}}
\newcommand{\LCE}{\mathsf{LCE}}

\newcommand{\Rp}{\mathbb{R}_{+}}
\newcommand{\Z}{\mathbb{Z}}
\newcommand{\Zp}{\mathbb{Z}_{+}}
\newcommand{\Oh}{O}
\newcommand{\tOh}{\tilde{\Oh}}
\newcommand{\dd}{\mathinner{..}}
\newcommand{\floor}[1]{\lfloor#1\rfloor}
\newcommand{\ceil}[1]{\lceil#1\rceil}

\newcommand{\A}{\mathtt{A}}
\newcommand{\B}{\mathtt{B}}
\newcommand{\C}{\mathtt{C}}
\newcommand{\D}{\mathtt{D}}
\newcommand{\De}{\mathbf{D}}
\newcommand{\Da}{\tilde{\mathbf{D}}}

\newcommand{\sub}{\subseteq}
\newcommand{\Bin}{\mathrm{Bin}}
\newcommand{\Zz}{\mathbb{Z}_{\ge 0}}

\renewcommand{\alph}{\mathrm{alph}}

\newcommand{\poly}{\mathrm{poly}}

\makeatletter
\newcases{lrdcases}{\quad}{$\m@th\displaystyle{##}$\hfil}{$\m@th\displaystyle{##}$\hfil}{\lbrace}{\rbrace}
\makeatother

\title{An Algorithmic Bridge Between Hamming and Levenshtein Distances} 

\author{Elazar Goldenberg}{Academic College of Tel Aviv-Yafo, Israel}{elazargo@mta.ac.il}{https://orcid.org/0000-0001-7993-3580}{}

\author{Tomasz Kociumaka}{Max Planck Institute for Informatics, Saarland Informatics Campus, Saarbr\"ucken, Germany}{tomasz.kociumaka@mpi-inf.mpg.de}{https://orcid.org/0000-0002-2477-1702}{Work mostly carried out while at the University of California, Berkeley, partly supported by NSF 1652303, 1909046, and HDR TRIPODS 1934846 grants, and an Alfred P. Sloan Fellowship.}

\author{Robert Krauthgamer}{Weizmann Institute of Science, Israel}{robert.krauthgamer@weizmann.ac.il}{} {Work partially supported by ONR Award N00014-18-1-2364, the Israel Science Foundation grant \#1086/18, and a Minerva Foundation grant, and by the Israeli Council for Higher Education (CHE) via the Weizmann Data Science Research Center.}

\author{Barna Saha}{University of California, San Diego, USA}{barnas@ucsd.edu}{https://orcid.org/0000-0002-6494-3839} {Partly supported by NSF CCF grants 1652303 and 1909046, and an HDR TRIPODS Phase II grant 2217058.}

\authorrunning{E.~Goldenberg, T.~Kociumaka, R.~Krauthgamer, and B.~Saha} 

\Copyright{Elazar Goldenberg, Tomasz Kociumaka, Robert Krauthgamer, and Barna Saha} 

\ccsdesc[500]{Theory of computation~Pattern matching}
\ccsdesc[500]{Theory of computation~Streaming, sublinear and near linear time algorithms}

\keywords{edit distance, Hamming distance, Longest Common Extension queries}

\nolinenumbers 

\ifArXiv
\hideLIPIcs 
\else 
\relatedversiondetails{Full Version}{https://arxiv.org/abs/ZZZ}
\fi

\EventEditors{Yael Tauman Kalai}
\EventNoEds{1}
\EventLongTitle{14th Innovations in Theoretical Computer Science Conference (ITCS 2023)}
\EventShortTitle{ITCS 2023}
\EventAcronym{ITCS}
\EventYear{2023}
\EventDate{January 10--13, 2023}
\EventLocation{MIT, Cambridge, Massachusetts, USA}
\EventLogo{}
\SeriesVolume{251}
\ArticleNo{92}

\begin{document}

\maketitle

\begin{abstract}
The edit distance between strings classically assigns unit cost to every character insertion, deletion, and substitution, whereas the Hamming distance only allows substitutions. 
In many real-life scenarios, insertions and deletions (abbreviated \emph{indels}) appear frequently but significantly less so than substitutions.
To model this, we consider substitutions being cheaper than indels, with cost $\frac1a$ for a parameter $a\geq 1$. 
This basic variant, denoted $\EDA$, bridges classical edit distance ($a=1$) with Hamming distance ($a\to\infty$), leading to interesting algorithmic challenges:
Does the time complexity of computing $\EDA$ interpolate between that of Hamming distance (linear time) and edit distance (quadratic time)?
What about approximating $\EDA$? 

We first present a simple deterministic exact algorithm for $\EDA$ and further prove that it is near-optimal assuming the Orthogonal Vectors Conjecture.
Our main result is a randomized algorithm computing a $(1+\epsilon)$-approximation of $\EDA(X,Y)$, given strings $X,Y$ of total length $n$ and a bound $k\ge\EDA(X,Y)$.
For simplicity, let us focus on $k\ge 1$ and a constant $\epsilon>0$;
then, our algorithm takes $\tOh(\frac{n}{a} + ak^3)$ time.
Unless $a=\tOh(1)$, in which case $\EDA$ resembles the standard edit distance, and for the most interesting regime of small enough $k$, this running time is sublinear in $n$. 

We also consider a very natural version that asks to find a \emph{$(k_I, k_S)$-alignment}, i.e., an alignment with at most $k_I$ indels and $k_S$ substitutions. 
In this setting, we give an exact algorithm and, more importantly, an $\tOh(\frac{nk_I}{k_S}+k_Sk_I^3)$-time $(1,1+\epsilon)$-bicriteria approximation algorithm.
The latter solution is based on the techniques we develop for $\EDA$ for $a=\Theta(\frac{k_S}{k_I})$, and its running time is again sublinear in $n$ whenever $k_I \ll k_S$ and the overall distance is small enough.

These bounds are in stark contrast to unit-cost edit distance, where state-of-the-art algorithms are far from achieving $(1+\epsilon)$-approximation in sublinear time, even for a favorable choice of $k$.
\end{abstract}

\section{Introduction}
Edit distance and Hamming distance are the two most fundamental measures of sequence similarity. 
The (unit-cost) edit distance, also known as the Levenshtein distance, of two strings $X$ and $Y$ is the minimum number of character insertions, deletions, and substitutions required to convert one string to the other, whereas the Hamming distance allows only substitutions (requiring $|X|=|Y|$).
From an algorithmic perspective, these two measures exhibit significantly different time complexity in terms of the input size $n=|X|+|Y|$. 
The Hamming distance can be computed exactly in linear time $O(n)$, and it admits a randomized $\epsilon n$-additive approximation in time $\Oh(\epsilon^{-1})$,
which implies a $(1+\epsilon)$-approximation in sublinear time when the distance is not too small. 
In contrast, assuming the Orthogonal Vectors Conjecture~\cite{Wil05}, no $\Oh(n^{2-\Omega(1)})$-time algorithm can compute the edit distance exactly~\cite{BI18}.
Recent developments in designing fast approximation algorithms~\cite{BJKK04,AKO10,AO12,CDGKS18,BR20,KS20b,GRS20} culminated in an $O(1)$-approximation in near-linear time~\cite{AN20}, but the existence of a truly subquadratic-time $3$-approximation still remains open.
Furthermore, despite many efforts~\cite{BCFN22a,BCFN22b,BEKMRRS03,GKS19,KS20a,GKKS22}, the best approximation ratio achievable in sublinear-time
ranges from polylogarithmic to polynomial in $n$ (depending on how large the true edit distance is).

The contrasting complexity landscape between edit and Hamming distance clearly indicates that substitutions are easier to handle than insertions and deletions (abbreviated \emph{indels}). 
Many real-world applications, for instance in computational biology, compare sequences based on edit distance, but its value is often dominated by the number of substitutions, as indicated by recent studies~\cite{ZG03,NMH19,CYB09,H18,M10}. Is it possible to design significantly faster algorithms for these scenarios with much more substitutions than indels? 
Such a bridge between Hamming and edit distances could also provide an explanation for why many heuristics for string comparison are fast on real-life examples.
(For another applied perspective, see~\cite{Medvedev22}.)

This motivates our study of a basic variant of the edit distance, denoted $\EDA$, where a substitution is significantly cheaper than an insertion or a deletion, and its cost is $\frac{1}{a}$ for a parameter $a\geq 1$. 
This simple variant bridges unit-cost edit distance ($a=1$) and Hamming distance ($a\ge |X|=|Y|$), raising basic algorithmic questions:
\emph{Does the time complexity of computing $\EDA$ interpolate between Hamming distance (linear time) and edit distance (quadratic time)?
How efficiently can one compute a $(1+\epsilon)$-approximation of $\EDA$? What speed-up is feasible as the parameter $a$ grows?}
We present the first series of results to answer these questions. 
As an illustrative example, if an intended application expects about $k$ indels and $k^2$ substitutions between the two strings, 
then one should set $a=k$ and use the $(1+\epsilon)$-approximation algorithm that we devise,
which runs in sublinear time $\tilde{O}(\frac{n}{k}+\text{poly}(k))$ for this parameter setting.
In contrast, the state-of-the-art sublinear-time algorithm for unit-cost edit distance~\cite{BCFN22a} only provides an $\Oh(\operatorname{polylog} k)$-factor approximation of the total number of edits, so it will likely produce an alignment with $\omega(k^2)$ indels, which blatantly violates the structure of alignments arising in the considered application.

One can achieve a stricter control on the number of indels and substitutions by imposing two independent bounds $k_I$ and $k_S$ on these quantities.
Our techniques can be easily adapted to the underlying problem, which we call the \emph{$(k_I, k_S)$-alignment} problem;
in particular, we obtain a sublinear-time $(1,1+\epsilon)$-bicriteria approximation of the combined cost $(k_I,k_S)$.
For $(k_I,k_S)=(k,k^2)$ as in the example above, the running time of our algorithm is still $\tOh(\frac{n}{k}+\poly(k))$, albeit the $\poly(k)$ term is slightly larger.

The weighted edit distance problem and the complementary highest-score alignment problem are widely used in applications
and discussed in multiple textbooks; see e.g.~\cite{A05,G97,JM09}. 
The simplest version involves two costs (for indels and substitutions, respectively)
and, up to scaling, is equivalent to $\EDA$.
Theoretical analysis usually focuses on the fundamental unit-cost setting, but many results extend to the setting of constant costs, e.g., they admit the same conditional lower bound~\cite{BK15} (and these problems are clearly equivalent from the perspective of logarithmic approximation).
We initiate a deeper investigation of how the costs of edit operations affect the complexity of computing the edit distance (exactly and approximately); going beyond the regime of $a=\Oh(1)$ reveals the gamut between Levenshtein distance ($a=1$) and Hamming distance ($a=n$, $k< 1$),
providing a holistic perspective on these two fundamental metrics.

\subsection{Our Contribution} 
We provide multiple results for both the $\EDA$ and $(k_I, k_S)$-alignment problems. 
Even though these are basic problems that seamlessly bridge the gap between Hamming and edit distances, surprisingly little is known about them.
Throughout, we assume that the algorithms are given $\Oh(1)$-time random access to characters of $X,Y$ and know the lengths $|X|,|Y|$. 

Our main results are approximation algorithms, but let us start with exact algorithms for the two problems.
We first observe that a simple adaptation of the approach of Landau and Vishkin~\cite{LV88, LMS98} 
computes $\EDA(X,Y)$ in time $O(n+k \min(n,ka))$, where $k\ge \EDA(X,Y)$.
\ifArXiv%
We then use similar techniques to derive an algorithm for the $(k_I, k_S)$-alignment problem. 
\fi%
For convenience, we state these results for decision-only problems; it is straightforward to accompany every YES answer with a corresponding alignment. 

\begin{restatable}{proposition}{thmexact}\label{thm:exact}
Given two strings $X,Y\in \Sigma^*$ of total length $n$, 
a cost parameter $a\in \Zp$, and a threshold $k\in \Rp$,
one can compute $\EDA(X,Y)$ or report that $\EDA(X,Y)>k$
in deterministic time $\Oh(n + k \min(n, ka))$.
\end{restatable}

\begin{restatable}{proposition}{biexact}\label{prp:bi}
Given two strings $X,Y\in \Sigma^*$ of total length $n$ and two integer thresholds $k_S,k_I >0$, one can decide whether there is a $(k_I,k_S)$-alignment or not (report YES or NO) in deterministic time $\Oh(n+k_Sk_I^2)$.
\end{restatable}

The existing $n^{2-\Omega(1)}$ lower bounds~\cite{BI18,BK15} do not shed light as to whether these running times are optimal.
Interestingly, we strengthen the result of~\cite{BK15} significantly to show that the bound of \cref{thm:exact} is indeed tight for the entire range of parameters $a,k\ge 1$, assuming the Orthogonal Vectors Conjecture~\cite{Wil05}.
Proving such a lower bound for the bicriteria version remains open.

\subparagraph*{Approximation Algorithm for $\EDA$.}
Our main results are $(1+\epsilon)$-approximation algorithms for the $\EDA$ and $(k_I, k_S)$-alignment problems.
Let us first formally state the gap versions (also known as the promise versions) of these problems. 
\begin{problem}[Approximate Bounded \EDA]\label{prob:approx}
Given two strings $X,Y\in \Sigma^*$ of total length $n$, 
a cost parameter $a\in \Zp$, 
a threshold $k\in \Rp$, and an
accuracy parameter $\epsilon \in (0,1)$,
report YES if $\EDA(X,Y)\le k$, NO if $\EDA(X,Y)>(1+\epsilon)k$,
and an arbitrary answer otherwise.
\end{problem}

\begin{problem}[Bicriteria Approximation]\label{prob:bi}
Given two strings $X,Y\in \Sigma^*$ of total length $n$,
two thresholds $k_I,k_S>0$,
and two approximation factors $\alpha,\beta\ge 1$, 
return YES if there is a $(k_I,k_S)$-alignment,
NO if there is no $(\alpha k_I, \beta k_S)$-alignment, 
and an arbitrary answer otherwise.
\end{problem}

Our aim is to solve the above problems using algorithms whose running time is truly sublinear for suitable parameter values.
This is in stark contrast to unit-cost edit distance, where state-of-the-art algorithms are far from achieving $(1+\epsilon)$-approximation in sublinear time, even for a favorable choice of parameters.%
\footnote{The smallest approximation ratio known to improve upon the running time $\Oh(n+k^2)$ of the exact and conditionally optimal algorithm~\cite{LV88} stands at $3+o(1)$~\cite{GRS20}. The approximation ratio currently achievable in sublinear time is polylogarithmic in $k$ (if $k < n^{1/4-\Omega(1)}$)~\cite{BCFN22a} or polynomial in $k$ (otherwise)~\cite{GKKS22}.}

\begin{restatable}{theorem}{thmmain}\label{thm:main}
One can solve \cref{prob:approx} correctly with high probability within time:
\begin{enumerate}
	\item $\tOh(\frac{n}{\epsilon^2 ak})$ if $k < 1$;\label{it:ham}
	\item $\tOh(\frac{n}{\epsilon^3 a} +  ak^3)$ if $1 \le k < \frac{n}{\epsilon a}$; and\label{it:complex}
	\item $\Oh(1)$ if $k \ge \frac{n}{\epsilon a}$.\label{it:trivial}
\end{enumerate}
\end{restatable}
\medskip

Cases~\ref{it:ham} and~\ref{it:trivial} are boundary cases;
in the former, there are only substitutions (Hamming distance), 
whereas, in the latter, $|X|$ and $|Y|$ must differ significantly when $\EDA(X,Y)>k$, which results in a trivial $(1+\epsilon)$-factor approximation. 
Our main technical contribution is Case~\ref{it:complex}, 
where we achieve sublinear running time for large enough $a$ and small enough $k$.%
\footnote{For small $a$, e.g., $a=1$, using our exact algorithm 
would improve upon the bound in Case~\ref{it:complex} and, in particular, reduce the dependency on $k$ from cubic to quadratic.
}
In the aforementioned applications~\cite{ZG03,NMH19,CYB09,H18,M10},
the indels are few in number and must be estimated with high accuracy;
in this case, $a$ should be set proportionally to the substitution-to-indel ratio, 
so that our $(1+\epsilon)$-approximation computes, in sublinear time, a highly accurate estimate of both the indels and the substitutions.

It is straightforward to solve \cref{prob:bi} with $\alpha=\beta=2+\epsilon$,
by simply applying \cref{thm:main} with $a=\frac {k_S} {k_I}$ and threshold $k=2k_I$. 
\ifArXiv%
However, the techniques we developed for \cref{thm:main} allow for a much stronger guarantee of $\alpha=1$ and $\beta=1+\epsilon$.
\fi%

\begin{restatable}{theorem}{biapx}\label{thm:bi}
One can solve \cref{prob:bi} with $\alpha=1$ and $\beta=1+\epsilon$  (for any given parameter $\epsilon \in (0,1)$)
correctly with high probability in $\tOh(\frac{nk_I}{\epsilon^3 k_S}+k_Sk_I^3)$ time.
\end{restatable}

Again, the running time is sublinear whenever $k_I \ll k_S$ and the overall distance is small.


\subsection{Notation}
A \emph{string} $X\in \Sigma^*$ is a finite sequence of characters from an \emph{alphabet} $\Sigma$.
The length of $X$ is denoted by $|X|$ and, for $i\in [0\dd |X|)$,\footnote{For $i,j{\in}\mathbb{Z}$, denote $[i\dd j]=\{k\in \mathbb{Z} : i \le k \le j\}$, $[i\dd j)=\{k\in \mathbb{Z} : i \le k < j\}$, $(i\dd j]=\{k\in \mathbb{Z} : i < k \le j\}$.} the $i$th character of $X$ is denoted by $X[i]$. 
A string $Y$ is a \emph{substring} of $X$ if $Y = X[i]X[i+1]\cdots X[j-1]$ for some $0\le i \le j \le |X|$;
this \emph{occurrence} of $Y$ is denoted by $X[i\dd j)$ or $X[i\dd j-1]$ and called a \emph{fragment} of $X$.
According to this convention, the empty string has occurrences $X[i\dd i)$ for $i\in [0\dd |X|]$.
We use $\HD$ to denote the Hamming distance between two strings and $\ED$ to denote the standard edit distance.

A key notion in our work is that of Longest Common Extension (LCE) queries, defined as follows
for indices $x\in [0\dd |X|]$ and $y\in [0\dd |Y|]$ in strings $X,Y\in \Sigma^*$:\footnote{Here and throughout, we implicitly assume the range $\ell\in [0\dd \min(|X|-x,|Y|-y)]$ needed to guarantee that $X[x\dd x+\ell)$ and $Y[y\dd y+\ell)$ are well-defined.}
\[
  \LCE(x,y) = \max\{\ell :  X[x\dd x+\ell)=Y[y\dd y+\ell)\}.
\]
After $\Oh(|X|+|Y|)$-time preprocessing, $\LCE$ queries can be answered in $\Oh(1)$ time~\cite{LV88,FFM00}.
This notion is often generalized to find the maximum length for which 
the corresponding substrings still have a small Hamming distance; formally,
\[
  \LCE_d(x,y) = \max\{\ell : \HD(X[x\dd x+\ell),Y[y\dd y+\ell))\le d\}.
\]
We further define $\LCE_{d,\epsilon}(x,y)$ as an arbitrary value between $\LCE_d(x,y)$ and $\LCE_{(1+\epsilon)d}(x,y)$; intuitively, this represents $\LCE_d(x,y)$ up to a $(1+\epsilon)$-approximation of the value of $d$.

\subsection{Technical Overview}\label{sec:overview}

\subparagraph*{Exact Algorithm for \boldmath $\EDA$.}
The algorithm behind \cref{thm:exact} is rather straightforward, but it serves as a foundation for subsequent results.
The naive way of computing $\EDA(X,Y)$ is to construct a dynamic-programming (DP) table
$T[x,y]=\EDA(X[0\dd x), Y[0\dd y))$.
Considering only entries with $|x-y|\le k$ (others clearly exceed $k$)
easily yields an $\Oh(n+nk)$-time algorithm.
If $ak \le n$, we achieve a better running time of $\Oh(n+ak^2)$ by generalizing the Landau--Vishkin algorithm~\cite{LV88,LMS98} to allow $a>1$. 
As explained next, this involves answering $\Oh(ak)$ $\LCE$ queries for each of $(2k+1)$ diagonals (which consist of entries $T[x,y]$ with fixed $y-x$).
The table $T$ is monotone along the diagonals, and thus one can construct ``waves'' of entries with a common value $v$.
This structure is conveniently described as another table that, 
for every possible cost $v\in \set{0,\frac{1}{a},\ldots, k-\frac{1}{a},k}$ 
and every possible shift (diagonal) $s\in [-k\dd k]$, contains
an entry $\De_v[s]$ defined as the furthest row $x$ in $T$ with $T[x,x+s]\le v$,
or equivalently, the maximum index $x$ such that $\EDA(X[0\dd x), Y[0\dd x+s))\le v$. 
To compute $\De_v[s]$, the algorithm uses three previously computed entries (for costs $v'<v$)
and then performs a single $\LCE$ query.
The running time of this algorithm is dominated by the construction and usage 
of a data structure answering $\LCE$ queries; see \cref{sec:ExactAlg} for details.

\subparagraph*{A Naive Sampling Algorithm.}
A natural way to speed up the previous algorithm at the expense of accuracy is to use sampling. 
As a first attempt, 
let the algorithm sample positions in $X$ at some rate $r\in(0,1)$
and compare each sampled $X[i]$ with $Y[i+s]$ for every $s\in [-k\dd k]$.
The algorithm then reports the minimum-cost alignment of the samples
(i.e., each sampled $X[i]$ is associated with exactly one shift $s$ and is matched to $Y[i+s]$),
where shift changes cost 1 per unit (i.e., changing $s$ to $s'$ costs $|s-s'|$)
and mismatches cost $\frac{1}{ar}$ each.
At best, we may hope to set $r=\Theta(\frac{1}{ka})$, 
the minimum needed for an optimal alignment with $\Theta(ka)$ substitutions~and~no~indels.

This approach faces two serious obstacles.
First, the query complexity is asymmetric: the algorithm queries, in expectation,
$\Oh(nr)$ positions in $X$ and $\Oh(knr)$ positions in~$Y$.
The query complexity into $Y$ must be improved, ideally to $\Oh(nr)$ as well.
Second, the hope for sampling rate $r = \Theta(\frac{1}{ka})$ is not realistic, 
because, in the case of $a=1$ (i.e., the standard edit distance),
this would distinguish between $\ED(X,Y)\le k$ and $\ED(X,Y)>(1+\epsilon)k$ using $\tOh(\frac{n}{k})$ samples, which is far beyond the reach of current techniques for edit distance estimation.

To overcome the first obstacle, we utilize approximate periodicity in our subroutine for answering approximate LCE queries, described later in this overview.
To tackle the second obstacle, we make sure that every observed edit is charged against approximately $\frac1r$ unobserved mismatches (substitutions), for a judicious choice of $r=\tilde O(\frac{1}{\epsilon a})$, as explained~next.

\subparagraph*{\boldmath Our $(1+\epsilon)$-Approximation Algorithm for $\EDA$.}
Let us modify our exact algorithm as follows. Instead of considering all possible values of $v$,
we enumerate only over $v\in\set{0, \epsilon,2\epsilon,\ldots, k}$,
where $\epsilon\ge \frac1a$ (otherwise, the exact algorithm already meets the requirements of \cref{thm:main}).
More precisely, for each diagonal $s$, we seek a value $\Da_v[s]$
that approximates $\De_v[s]$ up to a $(1+\epsilon)$-factor slack in~$v$.
In other words, $\Da_v[s]$ is between $\De_v[s]$ and $\De_{v(1+\epsilon)}[s]$;

Suppose that we have already computed $\Da_{v'}[s]$ for all $v'<v$.
Then, we can compute $\Da_v[s]$ using a single $\LCE_{\epsilon a,\epsilon}$ query.
Intuitively, this works well because these queries 
$(1+\epsilon)$-approximate the number of substitutions in the optimal alignment, 
except for up to $\epsilon a$ additional substitutions (of total cost $\epsilon$)
per indel (of cost $1$) in the optimal alignment; see \cref{sec:approx} for details.
Thus, the main engine of our algorithm is the following new tool:\footnote{%
	To be precise, our algorithm makes $\Oh(\epsilon^{-1}k^2)$ queries of \cref{prob:LCE} with $d=\Theta(\epsilon a)$ and $w=k$. A single $\LCE_{\epsilon a,\epsilon}(0,0)$ query may already require accessing $\Omega(\frac{n}{\epsilon^3 a})$ characters, so the best time we could hope for is $\tOh(\frac{n}{\epsilon^3 a}+\epsilon^{-1} k^2)$.
	If $\epsilon = \Theta(\frac{1}{a})$, this cannot improve upon our exact algorithm of~\cref{thm:exact}.
}

\begin{problem}[Approximate LCE queries]\label{prob:LCE}
Given strings $X,Y$ of total length $n$, a threshold $d\in \Zp$, a parameter $\epsilon\in \Rp$,
and a width $w\in \Zp$, build a data structure that efficiently computes $\LCE_{d,\epsilon}(x,y)$
for any  $x\in [0\dd |X|]$ and $y\in [0\dd |Y|]$ such that $|x-y|\le w$.
\end{problem}

\subparagraph*{Answering \boldmath $\LCE_{d,\epsilon}$ Queries.}
The following theorem captures our key technical innovation.

\begin{restatable}{theorem}{prpLCE}\label{prp:LCE}
After $\tOh(\frac{n}{\epsilon^2d})$-time randomized preprocessing (successful w.h.p.),
the queries of \cref{prob:LCE} can be answered deterministically 
(with no further randomness) in~$\tOh(dw)$~time.
\end{restatable}
Previous work on edit distance~\cite{GKS19,KS20a,BCR20,GRS20,BCFN22a} developed data structures for approximate $\LCE$ queries in a weaker version allowing any value between $\LCE_0(x,y)$ and $\LCE_d(x,y)$. 
As discussed in subsequent paragraphs, these results rely on identifying perfectly periodic string fragments,
whereas, in our case, periodicity is approximate and holds up to $O(dw)$ mismatches. 

In this overview, we focus on a simplified problem that already necessitates the main novel ideas behind \cref{prp:LCE}.
Namely, we consider the task of distinguishing, for any given $s\in [0\dd w)=[0\dd |Y|-|X|]$,
whether $\HD(X, Y[s\dd s+|X|))$ is $\le d$ or $>1.1d$.
To see that this is a special case of \cref{prob:LCE},
observe that an $\LCE_{d,0.1}(0,s)$ query must return $|X|$ 
if $\HD(X, Y[s\dd s+|X|))\le d$ and a value strictly smaller than $|X|$
if $\HD(X, Y[s\dd s+|X|))> 1.1d$.

In this setting, our goal is to answer queries in $\tOh(dw)$ time after $\tOh(\frac{n}{d})$-time preprocessing.
In particular, we want to answer all the $w$ possible queries while probing $\tOh(\frac{n}{d} + dw^2)$ characters in total.
For comparison, let us consider as a baseline the guarantees achieved using simple techniques.
A very naive approach would be to answer each query independently by sampling characters of $X$ with rate $\tOh(\frac{1}{d})$ and comparing each sampled character $X[x]$ with the aligned character $Y[x+s]$; this solution probes $\tOh(\frac{nw}{d})$ characters in total.
Another idea, employed in~\cite{BCR20}, could be to sample the characters of $X$ and $Y$ with the same rate $\tOh(\frac{1}{\sqrt{d}})$ and,
for each $s\in [0\dd w)$, compare $X[x]$ with $Y[x+s]$ whenever both characters have been sampled simultaneously, which happens with probability $\tOh(\frac{1}{d})$; this solution probes $\tOh(\frac{n}{\sqrt{d}})$ characters in total.
Overall, we conclude that the simple techniques result in an effective sampling rate much larger than $\frac{1}{d}$,
while our goal is to achieve the optimal rate $\tOh(\frac{1}{d})$ (necessary already to answer a single query)
at the price of a moderate additive cost.

The first step of our solution is to eliminate shifts $s\in [0\dd w)$
with very large distance $\HD(X, Y[s\dd s+|X|))$.
For this, we use the naive approach to distinguish between Hamming distances $\le dw$ and $>2dw$.
This step costs $\tOh(\frac{n}{dw})$ time per query, which is $\tOh(\frac{n}{d})$ in total.
If our filtering leaves us with at most one candidate shift $\tilde{s}\in [0\dd w)$, the answer for this shift is precomputed naively by sampling.

The interesting case is when at least two shifts $\tilde{s}_1\neq \tilde{s}_2$ remain. 
A key observation is that the strings $X$ and $Y$ are then almost periodic with period $p:=|\tilde{s}_1-\tilde{s}_2|$;
formally, $X[x]=X[x-p]$ holds for all but $\Oh(dw)$ positions $x\in [p\dd |X|)$, and similarly for $Y$.
In order to streamline notation, in this overview, we show how to exploit this structure for $p=1$.

In this case, most positions $x$ in $X$ 
have \emph{uniform contexts} of any fixed length $c$, i.e., $X[x\dd x+c)=Y[y\dd y+c)=\mathtt{A}^c$ holds 
for some $\mathtt{A}\in \Sigma$ and all $y=x+s$ across $s\in [0\dd w)$. 
Intuitively, our goal is to quickly process uniform regions
and spend our additional budget of $\tOh(dw)$ time per query
on the $\Oh(cdw)$ positions with non-uniform contexts (each periodicity violation affects $\Oh(c)$ contexts).
Intuitively, uniform regions are much easier to handle due to the following observation: If we sample two strings $U,V\in \Sigma^m$ with an appropriate rate $r=\tOh(\frac{2}{d})$, and this sampling reveals only $\mathtt{A}$s, then $\HD(U,V)\le \HD(U,\mathtt{A}^m)+\HD(V,\mathtt{A}^m)\le \frac{d}2 + \frac{d}{2} \le d$ holds with high probability.
This argument does not require any synchronization between samples,
so we can reuse the same sample for all shifts.

The remaining challenge is that a rate-$\tOh(\frac{1}{d})$ sample cannot perfectly identify uniform and non-uniform regions. Thus, we need to design a sampling mechanism that provides an unbiased estimate of $\HD(X, Y[s\dd s+|Y|))$ 
yet allows skipping regions that look uniform from the perspective of a rate-$\tOh(\frac{1}{d})$ sampling.
For this, we use two random subsets $S_X \subseteq [0\dd |X|)$ and $S_Y \subseteq [0\dd |Y|)$, where each index is picked independently with rate $r=\tOh(\frac{1}{d})$.
Every mismatch $X[x]\ne Y[y]$ is given two chances for detection: when $x\in S_X$ or $y\in S_Y$.
In order to avoid double counting,
the mismatch at position $x\in S_X$ will be counted if and only if 
$X[x]$ has fewer occurrences than $Y[y]$ within the contexts
$X[x\dd x+c)$ and $Y[y\dd y+c)$ of length $c=\tOh(\frac{1}{r})$.
In this case, with high probability, the sampling reveals a character other than $X[x]$ within the contexts,
and thus the algorithm classifies $x$ as \emph{non-uniform}.
For positions $x\in S_X$ classified as non-uniform, our query algorithm explicitly compares $X[x]$ to $Y[y]$ (for $y=s+x$)
and, upon detecting a mismatch, reads the entire contexts $X[x\dd x+c)$ and $Y[y\dd y+c)$ to verify 
whether indeed $X[x]$ is less frequent than $Y[y]$ within these contexts.
We can afford both steps because the expected number of positions $x\in S_X$ with non-uniform contexts is $\Oh(rcdw)=\tOh(dw)$ and the expected number of detected mismatches $X[x]\ne Y[y]$ is $\Oh(rd)$.

For a complete proof of \cref{prp:LCE}, without the simplifying assumptions, see \cref{sec:LCE}.
The techniques employed to eliminate these assumptions build upon~\cite{KS20a}, where the simplified problem
involves deciding, for every $s\in [0\dd w)$, whether $\HD(X,Y[s\dd s+|X|))$ is $0$ or $>d$. 
Unlike in our setting, that auxiliary problem allows eliminating a shift $s$ upon discovery of a single mismatch $X[x]\ne Y[x+s]$.
This leads to a relatively straightforward $\tOh(\frac{n}{d}+w)$-time solution~\cite{KS20a,BCFN22a}, which we sketch next to facilitate easy comparison with our new method.
First, exact pattern matching is used to filter out shifts $s$ such that $X[0\dd 2w)\ne Y[s\dd s+2w)$
If at most one shift $\tilde{s}$ remains, then a naive $\tOh(\frac{n}{d})$-time solution can be used to estimate $\HD(X,Y[\tilde{s}\dd \tilde{s}+|X|))$.
On the other hand, if there are at least two shifts $\tilde{s}_1\ne \tilde{s}_2$ left, then $p:=|\tilde{s}_1-\tilde{s}_2|$ is a period of $X[0\dd 2w)$. 
A naive rate-$\tOh(\frac{1}{d})$ sampling is used to check whether this period extends to the entire $X$ and the relevant portion of $Y$. 
If this procedure does not identify any violation of the period, then all the remaining shifts satisfy $\HD(X,Y[s\dd s+|X|))\le d$ with high probability.
Otherwise, another call to exact pattern matching (for substrings near the period violation) eliminates all but at most one of the remaining shifts.

\subparagraph*{Algorithms for the \boldmath $(k_I,k_S)$-Alignment Problem.}
Our algorithms computing $\EDA$ (exactly or approximately) can be easily adapted to solve the $(k_I,k_S)$-alignment problem.
The only modification required is to add another dimension 
to the $\De$ and $\Da$ tables so that the algorithm separately keeps track of the number of indels and substitutions (as opposed to their total weight).
In particular, $\De_{v_I,v_S}[s]$ is the maximum index $x$ such that $X[0\dd x)$ and $Y[0\dd x+s)$ admit a $(v_I,v_S)$-alignment,
and $\Da_{v_I,v_S}[s]$ approximates $\De_{v_I,v_S}[s]$ up to a small slack in the value $v_S$.
The exact algorithm uses $\Oh(k_S k_I)$ $\LCE$ queries for each of the $2k_I+1$ main diagonals,
whereas the approximation algorithm asks $\Oh(\frac{k_S k_I}{d})$ $\LCE_{d,\epsilon}$ queries per diagonal for a carefully chosen parameter $d=\Theta(\frac{\epsilon k_S}{k_I})$; 
\ifArXiv%
see \cref{app:bic} for details.
\fi%

\subparagraph*{Tight Lower Bound for Exact Computation of $\EDA$.}
Our lower bound essentially proves that the $\Oh(n+k\cdot \min(n,ak))$ running time in \cref{thm:exact}
is, up to subpolynomial factors, point-wise optimal as a function of $(n,a,k)$. 
We formalize this delicate statement as follows:

\begin{restatable}{theorem}{thmlb}\label{thm:lb}
	Consider sequences $(a_n)_{n=1}^{\infty}$ and $(k_n)_{n=1}^{\infty}$ with entries $a_n,k_n\in [1\dd n]$ computable in $\poly(n)$ time.
	Unless the Orthogonal Vectors Conjecture fails, there is no algorithm that, for some fixed $\epsilon > 0$, every $n\in \Zp$, and all strings $X,Y$ with $|X|+|Y|\le n$,
	in $\Oh((n + k_n\cdot \min(n, a_n k_n))^{1-\epsilon})$ time computes $\ED_{a_n}(X,Y)$ or reports that $\ED_{a_n}(X,Y)> k_n$.
\end{restatable}

We prove our conditional lower bound in two steps. First, we show that, for any $a\in [1\dd n]$,
$\Omega(n^{2-\epsilon})$ time is necessary to compute $\EDA(X,Y)$ with $|X|+|Y|\le n$.
We then build upon this construction to prove a lower bound on computing $\EDA(X,Y)$ under a guarantee that $\EDA(X,Y)\le k$ for a given $k=o(n)$.
In the first step, we follow the approach of Bringmann and Künnemann~\cite{BK15},
who proved the desired lower bound for $a=\Oh(1)$.
Unfortunately, when translated to larger $a$, their arguments only exclude $\Oh(n^{2-\epsilon}a^{-6})$-time algorithms.
The tool we adapt from~\cite{BK15} is a generic reduction from the Orthogonal Vectors problem to the problem of computing an abstract string similarity measure. This framework is formulated in terms of an \emph{alignment gadget} that needs to be supplied for the measure in question. 
Unlike~\cite{BK15}, instead of working directly with $\EDA(X,Y)$, we consider an \emph{asymmetric} 
measure $\D^+(X,Y):=\EDA(\$^{|Y|}\cdot X \cdot \$^{|Y|},Y)-|XY|$, where $\$$ is a character not present in $Y$.
Even for $a=\Oh(1)$, this trick greatly simplifies the proof in~\cite{BK15} at the cost of increasing the alphabet size from $2$ to $9$.
As a result, for each parameter $a$, every instance of the Orthogonal Vectors problem
produces a pair of strings $(X,Y)$ and a threshold $k$ such that $\EDA(X,Y)\le k$ if and only if the instance is a YES instance. However, since $\EDA(X,Y)=\Theta(|X|+|Y|)$ for instances obtained through $\D^+(\cdot,\cdot)$,
this construction alone provides no information on how the running time depends on $k$.
For this, we express an instance of the Orthogonal Vectors problem as an OR-composition of smaller instances,
and, for each of them, we construct a pair of strings $(X_i,Y_i)$ such that the original instance is a YES-instance if and only if $\min_i \EDA(X_i,Y_i)\le k_{\min}$. An appropriate gadget combines the pairs $(X_i,Y_i)$ into a single pair $(X,Y)$ with $\min_i \EDA(X_i,Y_i)\le k_{\min}$ if and only if $\EDA(X,Y)\le k$ for some $k=\Theta(\frac{n}{a})$, yielding an $\Omega((nk)^{1-\epsilon})$-time lower bound for this case. The lower bound of $\Omega((ak^2)^{1-\epsilon})$ for $k=\Oh(\frac{n}{a})$ follows because the problem in question does not get easier as we increase $n$ while preserving $a$ and $k$;
see \cref{sec:LB} for details.

\section{Computing \boldmath $\EDA$ Exactly}\label{sec:ExactAlg}

\thmexact* 

Define a DP table $T$ with $T[x,y]=\EDA(X[0\dd x), Y[0\dd y))$ for $x\in [0\dd |X|]$ and $y\in [0\dd |Y|]$. 
Then, $T[0,0]=0$ and the other entries can be computed according to the following formula (each argument of $\min$ is included only if the corresponding condition holds):
\[T[x,y] :=  \min \begin{lrdcases}
	T[x-1,y]+1& \text{if }x>0\\
	T[x,y-1]+1& \text{if }y>0\\
	T[x-1,y-1]+\tfrac{1}{a} \cdot \mathbb{1}_{X[x-1]\ne Y[y-1]}  & \text{if }x>0\text{ and }y>0
\end{lrdcases}.\]

One algorithm runs in time $\Oh(n+nk)$ and follows~\cite{Ukkonen85}.
It computes all entries $T[x,y]$ that satisfy $T[x,y]\le k$; this is correct because $T[x,y]$ only depends on entries $T[x',y']\le T[x,y]$. 
As $T[x,y] \ge |x-y|$, this algorithm considers $\Oh(n+nk)$ entries,
each computed in $\Oh(1)$ time.

Another algorithm runs in time $\Oh(n+ak^2)$ and follows~\cite{LV88}, as explained in Section~\ref{sec:overview}.
It is implemented as~\cref{alg:exact}, where all values $\De_v[s]$ are implicitly initialized to $-\infty$.

\begin{algorithm}[htb]
	\ForEach{$v\in \{0,\frac{1}{a},\frac{2}{a},\ldots,\frac1a\floor{ak}\}$}{
		\ForEach{$s\in [-\floor{v}\dd\floor{v}]$}{
			$\De'_v[s] \gets \min(|X|,|Y|-s,\max(\De_{v-1}[s-1],\De_{v-\frac {1}{a}}[s]+1,\De_{v-1}[s+1]+1))$\;
			\lIf{$s=0$}{$\De'_v[s] \gets \max(\De'_v[s],0)$}
			$\De_v[s] \gets \De'_v[s] + \LCE(\De'_v[s],\De'_v[s]+s)$\;
		}
		\lIf{$\De_v[|Y|-|X|]=|X|$}{\Return{$v$}}
	}
	\Return{``$>k$''}\;
	
	\caption{Exact algorithm for $\EDA$}\label{alg:exact}
	
\end{algorithm}

The correctness of \cref{alg:exact} follows from \cref{lem:exact} and the running time is proportional to the number of $\LCE$ queries, which is $\Oh(ak^2)$, plus the $\Oh(n)$ construction time of an $\LCE$ data structure~\cite{LV88,FFM00}.

\begin{restatable}{lemma}{lemexact}\label{lem:exact}
	For every $v\in \{0,\frac1a,\ldots,\frac1a\floor{ak}\}$, $x\in [0\dd |X|]$, and $y\in [0\dd |Y|]$, we have $T[x,y]\le v$ if and only if $\De_v[y-x]\ge x$.
\end{restatable}
\begin{proof}
Let us first prove, by induction on $v$, that $\De_v[y-x]\ge x$ if $T[x,y]\le v$.
Fix an optimal $\EDA$-alignment between $X[0\dd x)$ and $Y[0\dd y)$ and consider the longest suffixes $X[x'\dd x)=Y[y'\dd y)$ aligned without any edit.
We shall prove that $\De'_v[y-x]=\De'_v[y'-x']\ge x'$ by considering four cases:
\begin{itemize}
	\item If $x'=y'=0$, then $\De'_v[y'-x']=\De'_v[0]\ge 0$.
	\item If $X[x'-1]$ is deleted, then $v \ge T[x',y'] = T[x'-1,y']+1$, and the inductive assumption yields $\De'_v[y'-x']\ge \De_{v-1}[y'-x'+1]+1\ge x'$.
	\item If $Y[y'-1]$ is inserted, then $v\ge T[x',y']= T[x',y'-1]+1$, and the inductive assumption yields $\De'_v[y'-x']\ge \De_{v-1}[y'-x'-1]\ge x'$.
	\item Otherwise, $X[x'-1]$ is substituted for $Y[y'-1]$; then, $v\ge T[x',y']= T[x'-1,y'-1]+\frac1a$, and the inductive assumption yields $\De'_v[y'-x']\ge \De_{v-\frac1a}[y'-x']+1\ge x'$.
\end{itemize}
Now, $\De_v[y-x]\ge x$ follows from $\De'_v[y'-x']\ge x'$ due to $\LCE(x',y')\ge x-x'$.
	
The converse implication is also proved by induction on $v$. 
We consider four cases depending on $x':=\De'_{v}[y-x]$ and $y' := x'+(y-x)$:
\begin{itemize}
	\item If $x'\le \De_{v-1}[y'-x'+1]+1$, then, by the inductive assumption, $T[x',y'] \le T[x'-1,y']+1 \le (v-1)+1$.
	\item If $x'\le \De_{v-1}[y'-x'-1]$, then, by the inductive assumption, $T[x',y'] \le T[x',y'-1]+1 \le (v-1)+1$.
	\item If $x'\le \De_{v-\frac1a}[y'-x']+1$, then, by the inductive assumption, $T[x',y'] \le T[x'-1,y'-1]+\tfrac1a\le (v-\tfrac1a)+\tfrac1a$.
	\item In the remaining case, we have $x'=y'=0$, and thus $T[x',y']=0 \le v$.
\end{itemize}
In all cases, $T[x',y']\le v$ implies $T[x,y]\le v$ because $\LCE(x',y')\ge x-x'$.
\end{proof}

\section{Approximating $\EDA$}\label{sec:approx}

In this section, we present our solution for \cref{prob:approx}.
Recall that the input consists of strings $X,Y$, a cost parameter $a\in \Z_+$, a threshold $k\in \mathbb{R}_+$, and an accuracy parameter $\epsilon\in(0,1)$,
and the task is to distinguish between $\EDA(X,Y)\le k$ and $\EDA(X,Y)>(1+\epsilon)k$.
Our procedure mimics the behavior of~\cref{alg:exact} with a coarser granularity of the costs.
It is implemented as~\cref{alg:approx}, where we assume that all values $\Da_v[s]$ are implicitly initialized to $-\infty$ and that $\epsilon a \in \Z_+$ (we will fall back to \cref{alg:exact} whenever $\epsilon a < 1$).

\begin{algorithm}[h]
\caption{Approximation Algorithm}\label{alg:approx}
	\ForEach{$v\in \{0,\epsilon,2\epsilon,\ldots,\epsilon\ceil{\epsilon^{-1}k}\}$}{
		\ForEach{$s\in [-\floor{v}\dd \floor{v}]$}{
			$\Da'_v[s] \gets  \min(|X|,|Y|-s,\max(\Da_{v-1}[s-1],\Da_{v-\epsilon}[s],\Da_{v-1}[s+1]+1))$\;
			\lIf{$s=0$}{$\Da'_v[s] \gets \max(\Da'_v[s],0)$}
			$\Da_v[s] \gets  \Da'_v[s] + \LCE_{\epsilon a,\epsilon}(\Da'_v[s],\Da'_v[s]+s)$\;
		}
		\lIf{$\Da_v[|Y|-|X|] = |X|$}{\Return{YES}}
	}
	\Return{NO}
\end{algorithm}

The following lemma, justifying the correctness of \cref{alg:approx}, is 
proved below through an appropriate adaptation of the arguments behind \cref{lem:exact}.
\begin{restatable}{lemma}{lemapx}\label{lem:apx}
For every $v\in \{0,\epsilon,2\epsilon,\ldots,\epsilon \ceil{\epsilon^{-1}k}\}$, $x\in [0\dd |X|]$, and $y\in [0\dd |Y|]$,
if $T[x,y]\le v$, then $\Da_v[y-x]\ge x$,
and if $\Da_v[y-x]\ge x$, then $T[x,y]\le  v(1+\epsilon + \epsilon^2) + \epsilon + \epsilon^2$.
\end{restatable}
\begin{proof}
	We start by proving the first implication inductively on $v$.
	Let us fix an optimum $\EDA$-alignment of $X[0\dd x)$ and $Y[0\dd y)$
	and consider the longest suffixes $X[x'\dd x)$ and $Y[y'\dd y)$ aligned with at most $\epsilon a$ substitutions and no indels.
	We shall prove that $\Da'_v[y-x]=\Da'_v[y'-x']\ge x'$ by considering four cases:
	\begin{itemize}
		\item If $x'=y'=0$, then $\Da'_v[y'-x'] = \Da'_v[0] \ge 0$.
		\item If $X[x'-1]$ is deleted, then $v \ge T[x',y'] = T[x'-1,y']+1$, and the inductive assumption yields
		$\Da'_{v}[y'-x']\ge \Da_{v-1}[y'-x'+1]+1\ge x'$.
		\item If $Y[y'-1]$ is inserted, then  $v \ge T[x,y]= T[x',y'-1]+1$, and the inductive assumption yields $\Da'_{v}[y'-x'] \ge \Da_{v-1}[y'-x'-1]\ge x'$.
		\item Otherwise, there are exactly $\epsilon a$ substitutions between $X[x'\dd x)$ and $Y[y'\dd y)$ (recall that $\epsilon a \in \Z$).
		Thus,  $v\ge T[x,y]=T[x',y']+\epsilon$, and the inductive assumption yields $\Da'_{v}[y'-x'] \ge \Da_{v-\epsilon}[y'-x']\ge x'$.
	\end{itemize}
	Now, $\Da_v[y-x]\ge x$ follows from $\Da'_v[y'-x'] \ge x'$ due to $\LCE_{\epsilon a}(x',y')\ge x-x'$.	
	
	The second implication is also proved by induction on $v$.
	We consider four cases depending on $x':=\Da'_{v}[y-x]$ and $y' := x'+(y-x)$:
	\begin{itemize}
		\item If $x' \le \Da_{v-1}[y'-x'+1]+1$, then, by the inductive assumption, $T[x',y'] \le T[x'-1,y']+1 \le  (v-1)(1+\epsilon + \epsilon^2) + \epsilon + \epsilon^2 +1 = v(1+\epsilon + \epsilon^2)$.
		\item If $x' \le \Da_{v-1}[y'-x'-1]$, then, by the inductive assumption,  $T[x',y'] \le T[x',y'-1]+1 \le  (v-1)(1+\epsilon + \epsilon^2) +\epsilon + \epsilon^2 + 1  = v(1+\epsilon + \epsilon^2)$.
		\item If $x'\le \Da_{v-\epsilon}[y'-x']$, then, by the inductive assumption, $T[x',y'] \le  (v-\epsilon)(1+\epsilon + \epsilon^2)  + \epsilon + \epsilon^2 < v(1+\epsilon + \epsilon^2)$.
		\item In the remaining case, we have $x'=y'=0$. Trivially, $T[x',y']=0 \le v(1+\epsilon + \epsilon^2)$.
	\end{itemize}
	In all cases, $T[x',y']\le v(1+\epsilon + \epsilon^2)$ implies $T[x,y]\le v(1+\epsilon + \epsilon^2) + \epsilon + \epsilon^2$ because $\LCE_{(1+\epsilon)\epsilon a}(x',y')\ge x-x'$.
\end{proof}

\begin{corollary}\label{cor:apx}
\cref{alg:approx} returns 
\begin{itemize}
	\item YES if $\EDA(X,Y)\le k$, and 
	\item NO if $\EDA(X,Y)>k(1+\epsilon+\epsilon^2)+2\epsilon+2\epsilon^2+\epsilon^3$.
\end{itemize}
Moreover, it can be implemented in $\tOh(\frac{n}{\epsilon^3 a}+ak^3)$ time if $k\ge 1$.
\end{corollary}
\begin{proof}
If $\EDA(X,Y)\le k$, we apply \cref{lem:apx} for $v = \epsilon \ceil{\epsilon^{-1}k}$, $x=|X|$, and $y=|Y|$.
We conclude that $\Da_v[|X|-|Y|]\ge |X|$, which means that the algorithm returns YES.\@ 
On the other hand, if the algorithm returns YES, then $\Da_v[|X|-|Y|]\ge |Y|$ holds for some $v\in \{0,\epsilon,2\epsilon,\ldots,\epsilon \ceil{\epsilon^{-1}k}\}$, which satisfies $v < k+\epsilon$. In this case, \cref{lem:apx} implies that \[\ED(X,Y)\le v(1+\epsilon + \epsilon^2) + \epsilon + \epsilon^2 < (k+\epsilon)(1+\epsilon + \epsilon^2) + \epsilon + \epsilon^2 = k(1+\epsilon+\epsilon^2) + 2\epsilon+2\epsilon^2 + \epsilon^3.\]

As for the running time, we observe that the number of $\LCE_{\epsilon a, \epsilon}(x,y)$ queries asked is $\Oh((1+k)(1+\epsilon^{-1}k))$ and that each of them satisfies $|x-y|\le k$. Consequently, the claimed running time follows from \cref{prp:LCE}.
\end{proof}

\thmmain*
\begin{proof}
		First, consider $k<1$. If $|X|\ne |Y|$, then we can safely return NO due to $\EDA(X,Y)\ge \big||X|-|Y|\big|$. Otherwise, we output YES if $\HD(X,Y)\le ak$ (which implies $\EDA(X,Y)\le k$)
		 and NO if $\HD(X,Y)>(1+\epsilon)ak$ (which implies $\EDA(X,Y)\ge \min(1,(1+\epsilon) k)$).
		Previous work for Hamming distance (e.g.,~\cite{HIM12}) lets us distinguish these two possibilities in $\tOh(\frac{n}{\epsilon^2 ak})$ time.

		Next, consider $1\le k  < \frac{n}{\epsilon a}$. If $\epsilon < \frac{7}{a}$, we use the algorithm of \cref{thm:exact},
		which either computes $\EDA(X,Y)$ or reports that $\EDA(X,Y)>k$. This is clearly sufficient to distinguish
		between $\EDA(X,Y)\le k$ and $\EDA(X,Y)>(1+\epsilon)k$ for any $\epsilon \ge 0$.
		The running time is $\Oh(n+ak^2)=\Oh(\frac{n}{\epsilon a} + ak^2)=\tOh(\frac{n}{\epsilon^3 a} + ak^3)$.
		The most interesting case is when $\epsilon \ge \frac{7}{a}$. We then run \cref{alg:approx} with the accuracy parameter 
		decreased to $\bar{\epsilon} := \frac1a\floor{\frac{\epsilon a}7}$; in particular, this guarantees $\bar{\epsilon}a\in \Z_+$.
		By \cref{cor:apx}, the algorithm returns YES if $\EDA(X,Y)\le k$,
		and, due to  $(1+\epsilon)k \ge (1+7\bar{\epsilon})k \ge (1+\bar{\epsilon}+ \bar{\epsilon}^2)k+2\bar{\epsilon} + 2\bar{\epsilon}^2+\bar{\epsilon}^3$,
		it returns NO if $\EDA(X,Y) >(1+\epsilon)k$.
		The running time is $\tOh(\frac{n}{\bar{\epsilon}^3 a}+ak^3)=\tOh(\frac{n}{\epsilon^3 a}+ak^3)$ due to $\bar{\epsilon}\ge \frac\epsilon{14}$.

		Finally, consider $k\ge \frac{n}{\epsilon a}$. We  return YES if $\big||X|-|Y|\big|\le k$
		and NO otherwise. This is correct due to $\big||X|-|Y|\big|\le \EDA(X,Y)\le \big||X|-|Y|\big| + \frac1a\min(|X|,|Y|)\le \big||X|-|Y|\big| + \epsilon k$.
\end{proof}

\section{Answering \texorpdfstring{\boldmath $\LCE_{d,\epsilon}$}{Approximate LCE} Queries}\label{sec:LCE}

In this section, we explain how to implement $\LCE_{d,\epsilon}$ queries. We focus on the decision version of these queries, asking to distinguish whether $\HD(X[x\dd x+\ell), Y[y\dd y+\ell))$ is $\le d$ or $> (1+\epsilon) d$. A naive way of performing such a test would be to count mismatches $X[x+s]\ne Y[y+s]$ at positions $s\in S$, where each $s\in [0\dd \ell)$ is sampled into $S$ independently with an appropriate rate $r$.
The resulting estimator,  $H:=|\{s \in S : X[x+s]\ne Y[y+s]\}|$, follows the binomial distribution $\Bin(h, r)$,
where $h=\HD(X[x\dd x+\ell), Y[y\dd y+\ell))$.
Thus, we can compare $H$ against $(1+\frac{\epsilon}3)rd$ to distinguish between $h\le d$ and $h\ge (1+\epsilon)d$:
By the Chernoff bound, if $h\le d$,
then \[\Pr[H \ge (1+\tfrac{\epsilon}3)rd]\le \exp(-\tfrac1{27}\epsilon^2 rd),\]
whereas if $h\ge (1+\epsilon) d$,
then \[\Pr[H < (1+\tfrac{\epsilon}3)rd] \le \Pr[H<(1-\tfrac{\epsilon}{3})(1+\epsilon)rd]
\le \exp(-\tfrac1{18}\epsilon^2 rd).\]
In particular, if $r=\Theta(\epsilon^{-2}d^{-1}\log n)$, then w.h.p.\ the query algorithm is correct and costs $\tOh(1+r\ell)=\tOh(1+\epsilon^{-2}d^{-1}\ell)$ time.

The same approach can be used in the setting resembling that of \cref{prp:LCE},
where queries need to be answered deterministically after randomized preprocessing.
\begin{fact}\label{fct:naive}
	There exists a data structure that, after randomized preprocessing of strings $X,Y\in \Sigma^*$ and a rate $r\in (0,1)$,
	given positions $x\in [0\dd |X|]$, $y\in [0\dd |Y|]$ and a length $\ell\in [0\dd \min(|X|-x,|Y|-y)]$,
	outputs a value distributed as $\Bin(\HD(X[x\dd x+\ell),Y[y\dd y+\ell)),r)$.
	Moreover, the answers are independent for queries with disjoint intervals $[x\dd x+\ell)$.
	With high probability, the preprocessing time is $\tOh(1+r(|X|+|Y|))$ and the query time is $\tOh(1+r\ell)$.
\end{fact}
\begin{proof}
At preprocessing time, we draw $S_X\sub [0\dd |X|)$ with each position sampled independently with rate $r$.
At query time, the algorithm counts $x'\in S_X\cap [x\dd x+\ell)$ such that $X[x']\ne Y[y-x+x']$.
It is easy to see that each mismatch $X[x+s]\ne Y[y+s]$, with $s\in [0\dd \ell)$, is included with rate $r$,
and these events are independent across the mismatches. 
Moreover, the sets $S_X\cap [x\dd x+\ell)$ across distinct ranges $[x\dd x+\ell)$ are independent,
so the resulting estimators are independent as well. 
\end{proof}

Our aim, however, is a query time that does not grow with $\ell$.
As shown in \cref{sec:almostPeriodic,sec:AuxProb}, $\tOh(dw)$ query time can be achieved after preprocessing $X[x\dd x+\ell)$ in $\tOh(w+r\ell)$-time.\footnote{Recall the assumption in \cref{prob:LCE} that $|x-y|\le w$ holds for all queries.} 
While preprocessing each fragment $X[x\dd x+\ell)$ is too costly, it suffices to focus on fragments such that 
$[x\dd x+\ell)$ is a \emph{dyadic interval} (where $\log \ell$ and $\frac{x}{\ell}$ are both integers);
this is because, for any intermediate value $m\in [0\dd \ell)$, the problem of estimating $\HD(X[x\dd x+\ell),Y[y\dd y+\ell))$ naturally reduces to the problems of estimating both $\HD(X[x\dd x+m),Y[y\dd y+m))$ and $\HD(X[x+m\dd x+\ell),Y[y+m\dd y+\ell))$.
Estimators following binomial distribution are particularly convenient within such a reduction: if $B\sim \Bin(h,r)$ and $B'\sim \Bin(h',r)$ are independent,
then $B+B'\sim \Bin(h+h',r)$. Unfortunately, our techniques do not allow estimating $\HD(X[x\dd x+\ell),Y[y\dd y+\ell))$
when this value is much larger than $d$. Thus, we formalize the outcome of our subroutines as \emph{capped} binomial variables define below.

\begin{definition}
	An integer-valued random variable $B$ is a \emph{$d$-capped binomial variable} with parameters $h$ and $r$, denoted $B\in \Bin_d(h,r)$, if it satisfies the following conditions:
	\begin{itemize}
		\item If $h\le d$, then $B \sim \Bin(h,r)$, i.e., $\Pr[B \ge x] = \Pr[\Bin(h,r)\ge x]$ for all $x\in \mathbb{Z}$.
		\item Otherwise, $B$ stochastically dominates $\Bin(h,r)$, i.e., $\Pr[B \ge x] \ge \Pr[\Bin(h,r)\ge x]$ for all $x\in \mathbb{Z}$.
	\end{itemize}
\end{definition}

Moreover, to allow a small deviation from the desired distribution (e.g., in the unlikely event the number of sampled positions is much larger than expected), for $\delta\in [0,1]$, we write $B\in \Bin_{d,\delta}(h,r)$,
if $B$ is at total variation distance at most $\delta$ from a variable $B'\in  \Bin_{d}(h,r)$. 
The capped binomial variables can be composed just like their uncapped counterparts.
\begin{observation}\label{obs:combine}
	If $B\in \Bin_{d,\delta}(h,r)$ and $B'\in \Bin_{d,\delta'}(h',r)$ are independent, then $B+B'\in \Bin_{d,\delta+\delta'}(h+h',r)$.
\end{observation}

As indicated above, the main building block of our solution to \cref{prob:LCE} is a component for estimating $\HD(X[x\dd x+\ell),Y[y\dd y+\ell))$ for fixed $x,\ell$ and varying $y\in [x-w\dd x+w]$.
This task can be formalized as a problem of estimating \emph{text-to-pattern} Hamming distances~\cite{CGKKP20},
that is, the distances between a pattern $X$ and length-$|X|$ fragments of a text $Y$.
\begin{problem}\label{prob:PM}
	Preprocess strings $X,Y\in \Sigma^*$, real numbers $\delta\in (0,\frac{1}{|X|})$ and $r\in (0,1)$, and an integer $d\ge r^{-1}$	into a data structure that, given a shift $s\in [0\dd |Y|-|X|]$,
	outputs a value distributed as $\Bin_{d,\delta}(\HD(X,Y[s\dd s+|X|)),r)$.
\end{problem}

In \cref{sec:almostPeriodic}, we study \cref{prob:PM} in an important special case when $X,Y$ are almost periodic. The general solution to \cref{prob:PM} is then presented in \cref{sec:AuxProb}, and 
we derive \cref{prp:LCE} in \cref{sec:actual}.

\subsection{Text-to-Pattern Hamming Distances for Almost Periodic Strings}\label{sec:almostPeriodic}
In this section, we consider the case when $X$ and $Y$ are almost periodic. Formally, this means that the data structure is additionally given (at construction time) integers $p$ and $m$ such that $\HD(X[0\dd |X|-p),\allowbreak X[p\dd |X|))+\HD(Y[0\dd |Y|-p),Y[p\dd |Y|))\le m$. 

The presentation of our data structure comes in four parts. We start with an overview, where we introduce some notation and provide intuition. This is followed by~\cref{alg:PM}, which gives a precise mathematical definition of the values reported by our data structure. Then, in \cref{lem:main}, we formalize the intuition and prove that the outputs of \cref{alg:PM} satisfy the requirements of \cref{prob:PM}, that is, their distributions follow $\Bin_{d,\delta}(\HD(X,Y[s\dd s+|X|)),r)$.
Finally, in \cref{lem:impl}, we provide an efficient implementation of \cref{alg:PM} and the complexity analysis of the resulting procedure.

The randomness in our data structure comes from two sets $S_X\subseteq [0\dd |X|)$ and $S_Y\subseteq [0\dd |Y|)$ with each element sampled independently with probability $r$.
At a first glance, it may seem that $S_X$ would be sufficient because 
$|\{x\in S_X : X[x]\ne Y[s+x]\}|$ satisfies the requirements of \cref{prob:PM} (as proved in \cref{fct:naive}).
However, these values cannot be computed efficiently.
To see this, suppose that $X$ and $Y$ consist only of the character $\mathtt{A}$,
except for a single position $y$ where $Y[y]=\mathtt{B}$. 
In this case, we would need to return $1$ if and only if $y-s\in S_X$;
this requires $\Omega(|Y|)$ preprocessing time (to discover $y$) or $\Omega(|S_X|)$ query time
(to check if $y=s+x$ for some $x\in S_X$). 
Thus, our strategy is more involved: we partition the set of mismatches $\{(x,y) : X[x]\ne Y[y]\}$
into two disjoint classes, $M_X$ and $M_Y$, and we return $|\{(x,y)\in M_X : x\in S_X\text{ and }y-x=s\}| + |\{(x,y)\in M_Y : y\in S_Y\text{ and }y-x=s\}|$. In other words, depending on its class, a mismatch in $X[x]\ne Y[y]$ is counted either if $x\in S_X$ or if $y\in S_Y$.
Our classification is determined by comparing the frequency of characters $X[x]$ and $Y[y]$ in proximity to the two underlying positions.
Intuitively, this is helpful because a rare character is unlikely to be discovered unless it happens to be sampled.
Formally, for each position $x\in [0 \dd |X|)$, we define its context $C_x=\{x, x+p, x+2p,\ldots, x+(c-1)p\}$,
where $c=\tOh(r^{-1})$ is sufficiently large, and the context $C_y$ of $y\in [0 \dd |Y|)$ is defined similarly.
Then, we count the occurrences of $X[x]$ and $Y[y]$ as $X[x']$ for $x'\in C_x$ and as $Y[y']$ for $y'\in C_y$, denoting the resulting numbers by $m_x$ and $m_y$, respectively (see Lines~\ref{ln:mx} and~\ref{ln:my} of \cref{alg:PM}).
We place the mismatch in $M_X$ if $m_x \le m_y$ and in $M_Y$ otherwise.

\begin{algorithm}[b!]
\SetKwBlock{Begin}{}{end}
\SetKwFunction{preprocess}{Build}
\SetKwFunction{query}{Query}
	$H \gets 0$\;
	\ForEach{$x\in [0\dd |X|)$}{\label{ln:loop}
		$y \gets  s+x$\;\label{ln:y}
		\lIf{$x\notin S_X$ \KwSty{and} $y\notin S_Y$}{\KwSty{continue}}\label{ln:sxsy}
		$c \gets  \ceil{r^{-1}\ln \delta^{-2}}$\;
		$C_x \gets  \{x,x+p,\ldots,x+(c-1)p\}$\;
		$C_y \gets  \{y,y+p,\ldots,y+(c-1)p\}$\;
		\If{$C_x\sub [0\dd |X|)$ \KwSty{and} $|\{X[x']: x'\in C_x\cap S_X\} \cup \{Y[y']: y' \in C_y\cap S_Y\}|=1$}{\label{ln:filter}\KwSty{continue}\;\label{ln:filter2}}
		\lIf{$X[x]=Y[y]$}{\KwSty{continue}}\label{ln:mismatch}
		$m_x \gets  |\{x'\in C_x\cap [0\dd |X|) : X[x']=X[x]\}|+|\{y'\in C_y\cap [0\dd |Y|) : Y[y']=X[x]\}|$\;\label{ln:mx}
		$m_y \gets  |\{x'\in C_x\cap [0\dd |X|) : X[x']=Y[y]\}|+|\{y'\in C_y\cap [0\dd |Y|) : Y[y']=Y[y]\}|$\;\label{ln:my}
		\lIf{($x\in S_X$ \KwSty{and} $m_x \le m_y$) \KwSty{or} ($y\in S_Y$ \KwSty{and} $m_y < m_x$)}{$H \gets  H+1$}\label{ln:increment}
	}
	\Return{$H$}\;
	\caption{Answering queries of \cref{prob:PM} in the almost periodic case.}\label{alg:PM}
\end{algorithm}

The benefit of this approach is that, if $x\in M_X$, then the contexts $C_x$ and $C_y$
(of total size $2c$) contain at least $c$ occurrences of characters other than $X[x]$.
In particular, at a small loss of total variation distance,
we may assume that the samples $S_X$ and $S_Y$ have revealed such a character distinct from $X[x]$,
and thus we need to read $Y[y]$ only conditioned on that event.
At the same time, if $X$ and $Y$ are almost periodic and the sought Hamming distance is small,
then the contexts $C_x$ and $C_y$ are typically uniform (meaning that all the $2c$ positions are occurrences of the same symbol). Hence, we can afford to investigate each location with non-uniform contexts and, in case a mismatch $X[x]\ne Y[y]$ is detected, to explicitly compute $m_x$ and $m_y$ in order to verify whether this mismatch indeed belongs to $M_X$ or $M_Y$.

\begin{lemma}[Correctness]\label{lem:main}
Given $s\in [0\dd |Y|-|X|]$,  \cref{alg:PM} outputs a random variable $H$ at total variation distance at most $\delta$
from $\Bin(h,r)$, where $h=\HD(X,Y[s\dd s+|X|))$.
\end{lemma}

\begin{proof}
We shall first prove that a modification of \cref{alg:PM} with Lines~\ref{ln:filter} and~\ref{ln:filter2} removed outputs $H\sim \Bin(h,r)$.
Let $P\sub [0\dd |X|)$ be the set of positions for which $H$ was incremented in Line~\ref{ln:increment}.
We claim that $P$ is a subset of $M:= \{x\in [0\dd |X|) : X[x]\ne Y[s+x]\}$
with each position sampled independently with rate $r$.
Due to Line~\ref{ln:mismatch}, we have $P\sub M$.
For fixed positions $x\in M$ and $y=s+x$, let us consider the (deterministic) values $m_x$ and $m_y$ as defined in Lines~\ref{ln:mx} and~\ref{ln:my}. 
Observe that, for every $x\in M$, we have $x\in P$ if and only if $(x\in S_X\text{ and } m_x \le m_y)\text{ or }(y\in S_Y\text{ and }m_x > m_y)$.
Moreover, these are probability-$r$ events independent across $x\in M$.
Thus, in the modified version of the algorithm, we indeed have $H\sim \Bin(h,r)$.

Next, we shall prove that, for every $x\in P$, the probability of losing $x$ due to Line~\ref{ln:filter} does not exceed $\delta^2$.
By the union bound, the total variation distance between $H$ and $\Bin(h,r)$ is then at most $\delta^2 |X|\le \delta$.
Consider a multiset $A :=  \{X[x']: x'\in C_x \cap [0\dd |X|)\} \cup \{Y[y']: y' \in C_y\cap [0\dd |Y|)\}$.
The filtering of Line~\ref{ln:filter} applies only if $|A|=2c$ and all the sampled characters in $A$ match.
If $m_x \le m_y$ and $x\in S_X$, then $X[x]$ is sampled yet none of the at least $c$ elements of $A$ distinct from $X[x]$
is sampled.
Similarly, if $m_y < m_x$ and $y\in S_Y$, then $Y[y]$ is sampled yet none of the at least $c$ elements of $A$ distinct from $Y[y]$ is sampled.
In either case, the probability of losing $x\in P$ is bounded by $(1-r)^{c} \le \exp(-rc) \le \delta^2$, as claimed.
\end{proof}

\begin{lemma}[Efficient implementation]\label{lem:impl}
\cref{alg:PM} can be implemented so that, after $\tOh(|S_X|+|S_Y|)$-time preprocessing,
the query time is $\tOh((m+p+h+r^{-1})\log^2\delta^{-1})$ with probability $1-\Oh(\delta)$,
where $h=\HD(X,Y[s\dd s+|X|))$.
\end{lemma}

\begin{proof}
While preprocessing $X$, we construct a data structure that, given a position $x\in [0\dd |X|)$ and a character $a\in \Sigma$, in $\tOh(1)$ time computes the smallest position $x'\in S_X\cap [x\dd |X|)$ such that $X[x']\ne a$ and $x'\equiv x \pmod p$ (if any such $x'$ exists). This preprocessing costs $\tOh(|S_X|)$ time
(recall that all our model assumes oracle access to characters of the input strings).
The string $Y$ is preprocessed in the same way in $\tOh(|S_Y|)$ time.

In the main loop of Line~\ref{ln:loop}, we process $x\in [0\dd |X|)$ grouped by the remainder $x \bmod p$,
starting from $x\in [0\dd p)$ and $y=s+x$.
For each group, instead of increasing $x$ and $y$ by $p$ each time,
we utilize the precomputed data structure to perform larger jumps.
Specifically, if an iteration terminates due to Line~\ref{ln:sxsy}, we proceed directly
to the subsequent state $(x',y')$ such that $x'\in S_X$ or $y'\in S_Y$ (if any).
Moreover, while executing Line~\ref{ln:filter}, we compute $a=X[x]$ (if $x\in S_X$)
or $a=Y[y]$ (otherwise) and determine the nearest sampled position $x'\in [x\dd |X|)\cap S_X$
with $X[x']\ne a$ and $x' \equiv x \pmod p$ and the nearest sampled position $y'\in [y\dd |Y|)\cap S_Y$
with $Y[y']\ne a$ and $y'\equiv y \pmod p$. Then, the condition of Line~\ref{ln:filter}
is satisfied if and only if $(c-1)p < \min(|X|-x,x'-x,y'-y)$.
In that case, we can increase $x$ and $y$ by $p\cdot \lceil \frac1p \min(|X|-x,x'-x,y'-y) -(c-1)\rceil$;
the number of such increases can be bounded by $\Oh(p+m)$ in total across all the remainders modulo $p$.

It remains to count positions reaching Lines~\ref{ln:mismatch} and~\ref{ln:mx}.
The condition $C_x\sub [0\dd |X|)$ is not satisfied for at most $pc$ positions $x$.
The condition $|\{X[x']: x'\in C_x\cap [0\dd |X|)\} \cup \{Y[y']: y' \in C_{y}\cap [0\dd |Y|)\}|=1$
is not satisfied for at most $(m+h)c$ positions $x$.
Hence, at most $(m+h+p)c$ positions may proceed to Line~\ref{ln:mismatch}.
However, due to the filtering of Line~\ref{ln:sxsy}, the actual number of positions
reaching Line~\ref{ln:mismatch} can be bounded by $\Bin((m+h+p)c,2r)$, which is $\Oh((m+h+p)\log\delta^{-1})$ with probability at least $1-\delta$.
Hence, Lines~\ref{ln:y}--\ref{ln:mismatch} can be implemented in $\tOh((m+h+p)\log\delta^{-1})$ time with probability at least $1-\delta$.

Furthermore, $X[x]= Y[y]$ is not satisfied for at most $h$ positions,
and, for these positions, computing $m_x$ and $m_y$ takes $\Oh(c)$ time.
Due to the filtering of Line~\ref{ln:sxsy},  the actual number of positions
reaching Line~\ref{ln:mx} is at most $\Bin(h,2r)$, which is $\Oh(\frac{h}{r}+\log \delta^{-1})$ with probability at least $1-\delta$. Hence, Lines~\ref{ln:mx}--\ref{ln:increment}
cost $\Oh(h \log \delta^{-1} + r^{-1}\log^2\delta^{-1})$ time with probability at least $1-\delta$.
Overall, the running time of any query is $\tOh((m+p+h+r^{-1})\log^2\delta^{-1})$ with probability at least $1-\Oh(\delta)$.
\end{proof}

\subsection{Text-to-Pattern Hamming Distances for Arbitrary Strings}\label{sec:AuxProb}

In this section, we design an efficient solution to \cref{prob:PM} (in the general case):

\begin{restatable}{proposition}{prppm}\label{prp:PM}
	Given an instance of \cref{prob:PM} with $w:=|Y|-|X|+1$ and $|X|r\ge dw$,
	after $\tOh(|X|r\log\delta^{-1})$-time randomized preprocessing,
	one can deterministically (with no further randomness) answer the queries of \cref{prob:PM} in $\tOh(dw\log^2 \delta^{-1})$ time.
\end{restatable}
\begin{proof}
	In the preprocessing, for each shift $s\in [0\dd w)$, we distinguish,
	correctly with probability at least $1-\frac{\delta}{w}$, the case of $\HD(X,Y[s\dd s+|X|))\le dw$
	from the case of $\HD(X,Y[s\dd s+|X|))>2dw$.
	This is implemented by sampling positions in $X$ at rate $\Theta(\frac{\log (w\delta^{-1})}{dw})$
	and comparing the sampled characters of $X$ against the aligned characters of $Y$.
	By the union bound, the tests return correct values (for all $s\in [0\dd w)$) with probability at least $1-\delta$. 
	The running time of this preprocessing step is $\Oh(w\cdot \frac{\log (w\delta^{-1})}{dw}\cdot (|X|+|Y|)) =\tOh(|X|r\log\delta^{-1})$ with probability at least $1-\delta$ (if this step takes too much time,
	we terminate the underlying computation and declare a preprocessing failure, which effectively means that $|X|$ is returned for all queries).
	
	If the test reports a small value $\HD(X,Y[\tilde{s}\dd \tilde{s}+|X|))$ for just one shift $\tilde{s}\in [0\dd w)$,
	then we memorize this shift and compute a value distributed as $\Bin(\HD(X,Y[\tilde{s}\dd \tilde{s}+|X|)),r)$ by sampling the characters $X[x]$ and $Y[x+\tilde{s}]$ at rate $r$. This costs $\Oh(|X|r\log \delta^{-1})$ time
	with probability at least $1-\delta$ (otherwise, we declare a preprocessing failure).
	At query time, we report the computed value for $\tilde{s}$.
	For each of the remaining shifts $s\ne \tilde{s}$, we are guaranteed that (with probability at least $1-\delta$) $\HD(X,Y[s\dd s+|X|))>dw\ge d$, and thus we can return $|X|$ as a value distributed according to $\Bin_{d,\delta}(\HD(X,Y[s\dd s+|X|)),r)$.
	
	The interesting case is when we detect (at least) two shifts $\tilde{s}_1\ne \tilde{s}_2$ satisfying $\HD(X,Y[s\dd \allowbreak s+|X|))\le 2dw$.
	We then conclude that $X$ and $Y$ are approximately periodic with period $p := |\tilde{s}_1-\tilde{s}_2|$.
	Formally, this means that $\HD(X[0\dd |X|-p),X[p\dd |X|))+\HD(Y[0\dd |Y|-p),\allowbreak Y[p\dd |Y|))\le 8dw+w=\Oh(dw)$
	holds with probability at least $1-\delta$.
	In this case,  at preprocessing time, we draw sets $S_X\subseteq [0\dd |X|), S_Y\subseteq [0\dd |Y|)$ with each element sampled independently with probability $r$. This costs $\Oh(|X|r\log \delta^{-1})$ time
	with probability at least $1-\delta$ (otherwise, we declare a preprocessing failure).
	  We apply the query procedure of \cref{alg:PM}, described in~\cref{sec:almostPeriodic}, imposing a hard limit $\tOh(dw\log^2\delta^{-1})$ on the query time.
	If the query algorithm exceeds the limit, we return a naive upper bound $H = |X|$.
	If $h\le d$, the probability of exceeding the limit is $\Oh(\delta)$ (this incorporates both $X,Y$ violating the periodicity assumption and \cref{alg:PM} taking a long time due to an unlucky choice of $S_X,S_Y$).
	Hence, imposing the limit increases the 
	total variation distance of $H$ and $\Bin(h,r)$ by $\Oh(\delta)$.
	Otherwise (if $h > d$), setting $H:= |X|$ in some states of the probability space may only increase $H$ in these states,
	and thus the resulting random variable stochastically dominates the original one. In particular,
	$H$ is at total variation distance at most $\delta$ from a variable stochastically dominating $\Bin(h,r)$.
	Combining these two cases, we conclude that, if we shrink $\delta$ by an appropriate constant factor, the resulting value $H$ satisfies the requirements of \cref{prob:PM} 
\end{proof}

\subsection{Proof of Theorem~\ref{prp:LCE}}\label{sec:actual}
In this section, we derive \cref{prp:LCE} from \cref{prp:PM}. 

\prpLCE*

\begin{proof}
Our solution uses an auxiliary parameter $r\in[\frac{1}{d},1]$ to be chosen later.
At preprocessing, we build the data structure of \cref{fct:naive} and, for each dyadic range $[x\dd x+\ell)\sub [0\dd |X|)$ of length $\ell\ge \frac{dw}{r}$, the component of \cref{prp:PM} for $X[x\dd x+\ell)$ and $Y[x-w\dd x+w+\ell)$ (with the latter fragment trimmed to fit within $Y$ if necessary) and $\delta = n^{-c-1}$ for a sufficiently large constant $c$.
While doing so, we make sure that all these components use independent randomness. 
The construction algorithm takes $\tOh(\ell r)$ time for each such dyadic range,
which sums up to $\tOh(nr)$ time across all considered ranges.
	
As a result, for $x\in [0\dd |X|]$, $y\in [0\dd |Y|]$, and $\ell\in [1\dd \min(|X|-x,|Y|-y)]$, 
if $[x\dd x+\ell)$ is a dyadic range and $|y-x|\le w$, then we can in $\tOh(dw)$ time 
compute a value $H\in \Bin_{d,n^{-c-1}}(h,r)$, where $h=\HD(X[x\dd x+\ell),Y[y\dd y+\ell))$. 
For $\ell< \frac{dw}{r}$, this follows from \cref{fct:naive},
whereas for $\ell \ge \frac{dw}{r}$, this is a consequence of \cref{prp:PM}.
For an arbitrary range $[x\dd x+\ell)$, a value $H\in \Bin_{d,n^{-c}}(h,r)$ can be derived by decomposing  $[x\dd x+\ell)$ into $\Oh(\log n)$ dyadic ranges and combining the results of the individual queries using \cref{obs:combine}.

With an appropriate $r$, we can then use $H$ to distinguish $\HD(X[x\dd x+\ell),Y[y\dd y+\ell])\le d$ and $\HD(X[x\dd x+\ell),Y[y\dd y+\ell])\ge (1+\epsilon)d$.
To see this, first suppose that $H\in \Bin_d(h,r)$. If $h\le d$, then $\Pr[H \ge (1+\frac{\epsilon}{3})rd]
\le \exp(-\frac1{27}\epsilon^2 rd)$ holds by the multiplicative Chernoff bound.
Moreover, if $h \ge (1+\epsilon) d$, then $\Pr[H < (1+\frac{\epsilon}{3})rd] \le \Pr[H < (1-\frac{\epsilon}{3})(1+\epsilon)rd] \le \exp(-\frac1{18}\epsilon^2 rd)$ holds by the multiplicative Chernoff bound.
Hence, by testing whether $H<(1+\frac{\epsilon}{3})rd$, we can distinguish between $h\le d$ and $h\ge (1+\epsilon) d$
with failure probability at most $\exp(-\frac1{27}\epsilon^2 rd)$.
Setting $r = \Theta(\epsilon^{-2}d^{-1} \log n)$, the failure probability can be bounded by $n^{-c}$.

Since we are only guaranteed that $H\in \Bin_{d,n^{-c}}(h,r)$, the failure probability increases to $2n^{-c}$.
Nevertheless, this yields a data structure that can distinguish in $\tOh(dw)$ time between $\HD(X[x\dd x+\ell),Y[y\dd y+\ell])\le d$ and $\HD(X[x\dd x+\ell),Y[y\dd y+\ell])\ge (1+\epsilon)d$ correctly with high probability.
Since our query algorithm is deterministic, this means that, with high probability, the randomized construction algorithm yields a data structure that produces correct answers for all $\Oh(n^3)$ possible queries.
With binary search over $\ell$, we derive a data structure answering $\LCE_{d,(1+\epsilon)d}(x,y)$ queries
correctly w.h.p.\ and in $\tOh(dw)$ time.
\end{proof}

\section{Lower Bound for Computing $\EDA$ Exactly}\label{sec:LB}

The lower bounds provided in this section are conditioned on the Orthogonal Vectors Conjecture~\cite{Wil05},
which we state below.

\begin{conjecture}[Orthogonal Vectors Conjecture~\cite{Wil05}]\label{conj:ov}
In the $\OV$ problem, the input is a set $V\sub \set{0,1}^d$ of $n$ vectors,
and the goal is to decide if there exist $u,v\in V$ with $\tuple{u,v} =0$.\footnote{Here, $\tuple{u,v}=\sum_{i=1}^d u_i\cdot v_i$ denotes the scalar product of $u$ and $v$.}

The Orthogonal Vectors Conjecture asserts that, for every constant $\epsilon>0$, 
there exists a constant $c\ge 1$ such that $\OV$ cannot be solved in $\Oh(n^{2-\epsilon})$-time
on instances with $d=c\log n$. 
\end{conjecture}

\subsection{Unbounded Distance}

Here, we use the setting of~\cite{BK15} to prove that the following auxiliary similarity measures are hard to compute exactly. This immediately yields analogous hardness for \EDA.

\begin{definition}
For strings $X,Y\in \Sigma^*$ and an integer $a\in \Zp$,
let $D_a(X,Y) :=  \EDA(X,Y)+|Y|-|X|$.
Moreover, let $D^+_a(X,Y) :=  D_a(\$^{|Y|}\cdot X \cdot \$^{|Y|},Y)$, where $\$\notin \Sigma$.
\end{definition}

Intuitively, for an edit-distance alignment between strings $X$ and $Y$,
the measure $D_a$ charges a character $Y[y]$ with cost $0$ if it is matched, cost $\frac1a$ if it is substituted, and cost $2$ if it is inserted; deletions in $X$ are free.
The measure $D^+_a$ is justified by the following fact, where $\alph(Z)\sub \Sigma$ denotes the set of letters occurring in $Z\in \Sigma^*$.
\begin{restatable}{fact}{fctany}\label{fct:any_extension}
Let $X,Y,L,R\in \Sigma^*$.  If $|L|,|R|\ge |Y|$ and $\alph(Y)\cap \alph(LR)=\emptyset$,
then $D^+_a(X,Y)=D_a(LXR,Y)$.
\end{restatable}
\begin{proof}
Due to $\alph(Y)\cap \alph(LR)=\emptyset$, we may assume without loss of generality that $L=\$^{|L|}$
and $R=\$^{|R|}$. Any edit-distance alignment between $LXR$ and $Y$ deletes at least $|L|-|Y|$ characters of $L$
and at least $|R|-|Y|$ characters of $R$,
so $D_a(LXR,Y) \ge D_a(\$^{|Y|}\cdot X \cdot \$^{|Y|}, Y)$.
On the other hand, we have $\EDA(LXR,Y) \le \EDA(LXR,\$^{|Y|}\cdot X \cdot \$^{|Y|}) + \EDA(\$^{|Y|}\cdot X \cdot \$^{|Y|}, Y) = |L|-|Y|+|R|-|Y| + \EDA(\$^{|Y|}\cdot X \cdot \$^{|Y|}, Y)$ by the triangle inequality,
so $D_a(LXR) = \EDA(LXR,Y)+|Y|-|LXR| \le \EDA(\$^{|Y|}\cdot X \cdot \$^{|Y|}, Y)+|LR|-2|Y|+|Y|-|LXR|
= \EDA(\$^{|Y|}\cdot X \cdot \$^{|Y|}, Y)+|Y|-(|X|+2|Y|)= D^+_a(X,Y)$.
\end{proof}

Bringmann and Künnemann~\cite{BK15} developed a generic scheme of reducing the Orthogonal Vectors problem
to the problem of computing a given string similarity measure.\footnote{The authors wish to thank Marvin Künnemann for clarifying a few points about this framework.}
This framework requires identifying \emph{types} in the input space (which is $\Sigma^*$ here)
and proving that the similarity measure admits an \emph{alignment gadget} and \emph{coordinate values}.

In our construction, the types are of the form $[0\dd \sigma)^n$ for $n,\sigma\in \Zp$ (we assume $\Sigma\sub \Zz$).

\begin{restatable}{lemma}{lemcoord}
	The similarity measure $D^+_a$ admits coordinate values. That is, there exist strings $\mathbf{0}_X,\mathbf{1}_X$ (of some type $t_X$) and $\mathbf{0}_Y,\mathbf{1}_Y$ (of some type $t_Y$) such that $D^+_a(\mathbf{1}_X,\mathbf{1}_Y) > D^+_a(\mathbf{0}_X,\mathbf{0}_Y)=D^+_a(\mathbf{0}_X,\mathbf{1}_Y)=D^+_a(\mathbf{1}_X,\mathbf{0}_Y)$.
\end{restatable}
\begin{proof}
	We set $\mathbf{0}_X =  01,\mathbf{1}_X = 00$ (of type $[0\dd 2)^2$),
	and $\mathbf{0}_Y = 0,\mathbf{1}_Y = 1$ (of type $[0\dd 2)^1$).
	It is easy to check that $\frac1a = D^+_a(\mathbf{1}_X,\mathbf{1}_Y) > D^+_a(\mathbf{0}_X,\mathbf{0}_Y)=D^+_a(\mathbf{0}_X,\mathbf{1}_Y)=D^+_a(\mathbf{1}_X,\mathbf{0}_Y)=0$.
	\end{proof}

\newcommand{\GA}{\mathsf{GA}}
An alignment gadget is a pair of functions $\GA_X$ and $\GA_Y$ parameterized by integers $n\ge m$ and types $t_X,t_Y$. The function $\GA_X$, given $n$ strings $X_1,\ldots,X_n$ of type $t_X$, produces a string $X$ of some fixed type (that only depends on $n,m,t_X,t_Y$), whereas the function $\GA_Y$, given $m$ strings $Y_1,\ldots,Y_m$ of type $t_Y$, produces a string $Y$ of some (other) fixed type. Both functions must admit $\Oh((n+m)(\ell_X+\ell_Y))$-time algorithms, where $\ell_X$ is the common length of strings in $t_X$ and $\ell_Y$ is the common length of strings in $t_Y$.
The main requirement relates the value $D^+_a(X,Y)$ to the values $D^+_a(X_i,Y_j)$ for $i\in [1\dd n]$ and $j\in [1\dd m]$.  Specifically, there must be a fixed value $C$ (depending only on $n,m,t_X,t_Y$) that satisfies the following conditions:
\begin{enumerate}
	\item Each $\delta \in [0\dd n-m]$ satisfies $D^+_a(X,Y)\le C + \sum_{j=1}^m D^+_a(X_{j+\delta},Y_j)$.
	\item There is a set $A=\{(i_1,j_1),\ldots,(i_k,j_k)\}\sub [1\dd n]\times [1\dd m]$ such that $i_1< \cdots < i_k$,
	$j_1< \cdots < j_k$, and $D^+_a(X,Y) \ge C + \sum_{(i,j)\in A} D^+_a(X_i,Y_j) + (m-|A|)\max_{i,j}D^+_a(X_i,Y_j)$.
\end{enumerate}

\begin{construction}\label{cons:ag}
	Let $n\ge m$ be positive integers and $t_X = [0\dd \sigma_X)^{\ell_X}$, $t_Y= [0\dd \sigma_Y)^{\ell_Y}$ be input types. Denote $\A = \sigma_Y$, $\B =\sigma_Y+1$, and $\ell=\ell_Y$.
	We set $\GA_X$ and $\GA_Y$ so that 
	\begin{align*}
		\GA_X(X_1,\ldots, X_n) :=  X & :=  \bigodot_{i=1}^n \left(\A^{2\ell}\cdot X_i\cdot \A^{\ell}\cdot \B^{\ell}\right),\text{ and}\\ 
		\GA_Y(Y_1,\ldots, Y_m) :=  Y & :=  \B^{n\ell} \cdot \bigodot_{j=1}^m \left(\A^{\ell}\cdot Y_j \cdot \B^{\ell}\right)\cdot \B^{n\ell},
	\end{align*}
	where $\bigodot$ and $\cdot$ denote concatenation.
\end{construction}
It is easy to check that both functions can be implemented in $\Oh((n+m)(\ell_X+\ell_Y))$ time and the output types are  $[0\dd \max(\sigma_X,\sigma_Y+2))^{n(\ell_X+4\ell_Y)}$ and $Y\in [0\dd \sigma_Y+2)^{(2n+3m)\ell_Y}$, respectively.
In the next two lemmas, we prove that \cref{cons:ag} satisfies the two aforementioned conditions
with $C =  \frac1a(n+m)\ell$.

\begin{restatable}{lemma}{lemsub}
	Each $\delta\in [0\dd n-m]$ satisfies $D^+_a(X,Y)\le \frac1a(n+m)\ell+ \sum_{j=1}^{m} D^+_a(X_{j+\delta},Y_j)$.
\end{restatable}
\begin{proof}
	Let us construct an edit distance alignment between $\$^{|Y|}\cdot X \cdot \$^{|Y|}$ and $Y$ so that the total cost charged to characters of $Y$ does not exceed the claimed upper bound.
	\begin{itemize}
		\item The prefix $\B^{(n-\delta)\ell}$ is matched against the prefix $\$^{(n-\delta)\ell}$ (at cost $\frac1a (n-\delta)\ell$).
		\item The subsequent $\delta$ blocks  $\B^{\ell}$ are matched against the $\B^{\ell}$ blocks within the phrases $\A^{2\ell}\cdot X_{i}\cdot \A^{\ell}\cdot \B^{\ell}$ for $i\in [0\dd \delta)$ (at cost $0$).
		\item Each phrase $\A^{\ell}\cdot Y_j \cdot \B^{\ell}$ with $j\in [1\dd m]$ is matched against the phrase $\A^{2\ell}\cdot X_{j+\delta}\cdot \A^{\ell}\cdot \B^{\ell}$ so that:
		\begin{itemize}
			\item The leading block $\A^{\ell}$ is matched against the leading block $\A^{\ell}$ (at cost $0$).
			\item The string $Y_j$ is matched against $\A^{\ell}X_{j+\delta} \A^{\ell}$ (at cost $D^+_a(X_{j+\delta},Y_j)$ by \cref{fct:any_extension}).
			\item The trailing block $\B^{\ell}$ is matched against the trailing block $\B^{\ell}$ (at cost $0$).
		\end{itemize}
		\item The subsequent $n-m-\delta$ blocks $\B^{\ell}$ are matched against the $\B^{\ell}$ blocks within the phrases $\A^{2\ell}\cdot X_{i}\cdot \A^{\ell}\cdot \B^{\ell}$ for $i\in (m+\delta \dd n]$ (at cost $0$).
		\item The suffix $\B^{(m+\delta)\ell}$ is matched against the suffix $\$^{(m+\delta)\ell}$ (at cost $\frac1a (m+\delta)\ell$).\qedhere
	\end{itemize}
	\end{proof}

\begin{restatable}{lemma}{lemslb}
	There exists a set $A=\{(i_1,j_1),\ldots,(i_k,j_k)\}\sub [1\dd n]\times [1\dd m]$ such that $i_1< \cdots < i_k$,
	$j_1< \cdots < j_k$, and $D^+_a(X,Y) \ge \frac1a(n+m)\ell+ \sum_{(i,j)\in A}^m D^+_a(X_i,Y_j) + (m-|A|)\max_{i,j}D_a(X_i,Y_j)$.
\end{restatable}
\begin{proof}
	Let us fix an optimal alignment between $\$^{|Y|}\cdot X \cdot \$^{|Y|}$ and $Y$.
	We construct a set $A\sub  [1\dd n]\times [1\dd m]$ consisting of pairs $(i,j)$ jointly satisfying the following conditions:
	\begin{itemize}
		\item At least one character of $Y_j$ is aligned against a character of $X_i$.
		\item No character of $Y_j$ is aligned against any character of $X_{i'}$ for any $i'\ne i$.
		\item No character of $X_i$ is aligned against any character of $Y_{j'}$ for any $j' < j$.
	\end{itemize}
	It is easy to see that $A$ forms a non-crossing matching, i.e., ordering its elements by the first coordinate is equivalent	to ordering them by the second coordinate.
	
	Let us analyze the contribution of each letter of $Y$ to $D^+_a(X,Y)$ with respect to any fixed optimal alignment
	between $\$^{|Y|}\cdot X \cdot \$^{|Y|}$ and $Y$. Let $c_\B$ denote the number of $\B$s in $X$ that are not matched with any $\B$ in $Y$.
	Since the number of $\B$s in $X$ and $Y$ is $n\ell$ and $(2n+m)\ell$, respectively, the contribution of $\B$s in $Y$ to the total cost is at least $\frac1a((n+m)\ell+c_\B)$. 
	
	If $(i,j)\in A$, then $Y_j$ is aligned against a subsequence of $X$ obtained by deleting $X_{i'}$ for all $i'\ne i$ and (possibly) some further characters (recall that all deletions in $X$ are free).
	In this case, the contribution of $Y_j$ is at least $D^+_a(X_i,Y_j)$ by \cref{fct:any_extension}.
	
	If no character of $Y_j$ is aligned against any character of any $X_i$, then all characters of $Y_j$ are deleted or substituted, meaning that the contribution of $Y_j$ is at least $\frac1a\ell$.
	If characters of $Y_j$ are aligned against characters of $X_{i}$ and $X_{i'}$ for two distinct $i<i'$,
	then all $\B$s between $X_i$ and $X_{i'}$ contribute to $c_\B$; this contribution is at least $\ell$, and it cannot be charged to any $Y_{j'}$
	with $j'\ne j$.
	Finally, if characters of both $Y_{j'}$ and $Y_{j}$ (with $j' < j$) are aligned to characters of the same $X_i$,
	then the block $\A^\ell$ preceding $Y_j$ contributes at least $\frac1a\ell$.
	Overall, we conclude that the total cost is at least $\frac1a\ell$ for each $j\in [1\dd m]$ with no $(i,j)\in A$ (this contribution is either directly charged to $\A^\ell Y_j$ or via $c_\B$).
	Hence, the total cost $D^+_a(X,Y)$ is at least $\frac1a(n+2m-|A|)\ell + \sum_{(i,j)\in A} D^+_a(X_i,Y_j)$.
	Due to $D^+_a(X_i,Y_j)\le\frac1a|Y_j| = \frac1a\ell$, this yields the claimed lower bound.
	\end{proof}

Consequently, the framework of~\cite{BK15} yields the following result. (An inspection of~\cite{BK15} reveals that the strings are produced from the coordinate values composed by recursive application of the alignment gadget. The tree of this recursion is of height 3, so the output strings are over an alphabet of size 8.)

\begin{corollary}\label{cor:lb}
	There is an $\Oh(nd)$-time algorithm that, given two sets $U,V\sub \{0,1\}^d$ of $n$ vectors each
	and an integer $a\in \Zp$, constructs strings $X,Y\in [0\dd 8)^*$ and a threshold $k$ such that $D^+_a(X,Y)\le k$ if and only if $\langle u,v \rangle =0$ for some $u\in U$ and $v\in V$. Moreover, $|X|$, $|Y|$, and $k$ depend only on $n$, $d$, and $a$.
\end{corollary}

\subsection{Bounded Distance}

Our next goal is to derive a counterpart of \cref{cor:lb} with $k\ll |X|+|Y|$
and $\ED_a$ instead of $D_a^+$. The following construction is at the heart of our reduction.

\begin{construction}\label{cons:bnd}
	Let $n,\ell_X,\ell_Y$ be positive integers and let $\ell = \ell_X+\ell_Y$. For a sequence of $n$ pairs $(X_i,Y_i)\in \Sigma^{\ell_X}\times \Sigma^{\ell_Y}$, we define strings 
	\[
		X = \A^{\ell_Y}\cdot  \bigodot_{j=1}^n (\B^{\ell}\cdot X_j\cdot \C^{\ell}\cdot \A^{\ell_Y})\qquad\text{and}\qquad
		Y = Y_1\cdot  \bigodot_{j=2}^{n} (\B^{\ell}\cdot \D^{\ell_X}\cdot \C^{\ell} \cdot Y_j).
	\]
\end{construction}

\begin{restatable}{lemma}{lemulb}\label{lem:ulb}
If $a\ge n$, then the strings $X,Y$ of \cref{cons:bnd} satisfy
\[ D_a(X,Y) = \tfrac{(n-1)\ell}{a} + \min_{i=1}^n D^+_a(X_i,Y_i).\]
\end{restatable}
\begin{proof}
	First, we prove an upper bound on $D_a(X,Y)$.
	Let us fix $i\in [1\dd n]$. We define an edit distance alignment between \[X = \bigodot_{j=1}^{i-1} (\A^{\ell_Y}\cdot \B^{\ell}\cdot X_j\cdot \C^{\ell}) \cdot \A^{\ell_Y}\cdot \B^{\ell}\cdot X_i \cdot \C^{\ell}\cdot \A^{\ell_Y} \cdot \bigodot_{j=i+1}^{n} (\B^{\ell}\cdot X_j \cdot \C^{\ell} \cdot \A^{\ell_Y})\] and \[Y= \bigodot_{j=1}^{i-1} (Y_j\cdot \B^{\ell}\cdot \D^{\ell_X}\cdot \C^{\ell}) \cdot Y_i \cdot \bigodot_{j=i+1}^{n} (\B^{\ell}\cdot \D^{\ell_X}\cdot \C^{\ell} \cdot Y_j)\] so that the total cost charged to characters of $Y$ does not exceed $\tfrac{(n-1)\ell}{a} + D^+_a(X_i,Y_i)$.
	\begin{itemize}
		\item For $j\in [1\dd i)$, the phrase $Y_j\cdot  \B^{\ell}\cdot \D^{\ell_X}\cdot \C^{\ell}$  is aligned with the phrase $\A^{\ell_Y}\cdot \B^{\ell}\cdot X_j\cdot \C^{\ell}$ using substitutions only (at cost $\frac1a(\ell_X+\ell_Y)$). 
		\item The block $Y_i$ is aligned against $\A^{\ell_Y} \cdot B^\ell \cdot X_i\cdot \C^{\ell}\cdot \A^{\ell_Y}$ (at cost $D^+_a(X_i,Y_i)$ by \cref{fct:any_extension}).
		\item For $j\in (i\dd n]$, the phrase $(\B^{\ell}\cdot \D^{\ell_X}\cdot \C^{\ell} \cdot Y_j)$ is aligned with the phrase $\B^{\ell}\cdot X_j \cdot \C^{\ell} \cdot \A^{\ell_Y}$ using substitutions only (at cost $\frac1a(\ell_X+\ell_Y)$). 
	\end{itemize}
	
	It remains to prove the lower bound on $D_a(X,Y)$.
	Let us fix an arbitrary alignment of $X$ and $Y$; we shall prove that its cost is at least as large as the claimed lower bound on $D_a(X,Y)$.  
	First, suppose that, for some $i\ne j$, a character of $X_i$ is aligned against a character of $Y_j$.
	Note that $X_i = X[(3i-2)\ell + \ell_Y \dd (3i-1)\ell)$ and $Y_j = Y[(3j-3)\ell\allowbreak \dd (3j-3)\ell+\ell_Y)$.
	If $i < j$, then aligning a character of $X_i$ against a character of $Y_j$ requires deleting at least 
	$(3j-3)\ell-(3i-1)\ell = (3j-3i-2)\ell \ge \ell$
	characters in the prefix $Y[0\dd (3j-3)\ell)$.
	Similarly, if $i > j$, then aligning a character of $X_i$ against a character of $Y_j$ requires deleting at least 
	$(|Y|-((3j-3)\ell+\ell_Y))-((|X|-(3i-2)\ell + \ell_Y)) = (3(n-j)\ell - (3n-3i+2)\ell) = (3i-3j-2)\ell \ge \ell$
	characters in the suffix of $Y[(3j-3)\ell+\ell_Y\dd |Y|)$.
	In either case, the alignment cost is at least $\ell \ge \frac{n\ell}{a}\ge  \frac{(n-1)\ell}{a} + \frac{1}{a}\ell_Y \ge \frac{(n-1)\ell}{a} + \min_{i=1}^n D^+_a(X_i,Y_i)$. 
	
	Thus, we may assume that no character of $Y_j$ is aligned against a character of $X_i$ for $i\ne j$.
	In this case, we bound from below the total cost charged to characters of $Y$.
	Note that the symbols $\D$ in $Y$ contribute at least $\frac{n-1}{a}\ell_X$ (because there are no $\D$s in $X$). If, for some $i\in [1\dd n]$, no character of $Y_i$ is aligned against a character of $X_i$, then $Y_i$ is charged at least $\frac1a\ell_Y$.
	Otherwise, $Y_i$ is charged at least $D^+_a(X_i,Y_i)$ by \cref{fct:any_extension}.
	However, if $i<j$ are subsequent indices such that a character of $Y_i$ is aligned against a character of $X_i$ and a character of $Y_j$ is aligned against a character $X_j$ then, among the $2\ell(j-i)$ characters $\B$ and $\C$ between $Y_i$ and $Y_j$, at most $\mathsf{LCS}((\B^\ell \C^\ell)^{j-i},(\C^\ell \B^\ell)^{j-i})=\ell(2(j-i)-1)$
	are matched, incurring a cost of at least $\frac{1}{a}\ell>\frac{1}{a} \ell_Y$ for the mismatched characters.
	Overall, $\D$s contribute at least $\frac{n-1}{a}\ell_X$, the smallest $i$ such that a character of $X_i$ is aligned against a character of $Y_i$ (if there is any) contributes at least $D^+_a(X_i,Y_i)$, each of the remaining $i\in [1\dd n]$ contributes at least $\frac{1}{a}\ell_Y$ (either directly or via $\B$s and $\C$s), for a total of $\frac{(n-1)\ell}{a} + \min_{i=1}^n D^+_a(X_i,Y_i)$.
	\end{proof}

These properties let us use \cref{cons:bnd} to generalize \cref{cor:lb} and finally prove \cref{thm:lb}.
\begin{restatable}{proposition}{prplb}\label{prp:lb}
	There is an $\Oh(nmd)$-time algorithm that, given two sets $U,V\sub \{0,1\}^d$ of $n$ vectors each
	and integers $a\in \Zp$, $m\in [1\dd n]$, constructs strings $X,Y\in [0\dd 12)^*$ and a threshold $K=\Oh(\max(\frac{m}{a}, \frac{1}{m})nd)$ such that $\EDA(X,Y)\le K$ if and only if $\langle u,v \rangle =0$ for some $u\in U$ and $v\in V$.
\end{restatable}
\begin{proof}
	We cover $U$ and $V$ with subsets $U_1,\ldots, U_m\sub U$ and $V_1,\ldots,V_m\sub V$ of size $\lceil \frac{n}{m}\rceil$.
	For any $i,j\in [1\dd m]$, we construct an instance $(X_{i,j},Y_{i,j},k_{i,j})$ using \cref{cor:lb}
	for $U_i,V_j\sub \{0,1\}^d$.
	Note that $|X_{i,j}|=\ell_X$, $|Y_{i,j}|=\ell_Y$, and $k_{i,j}$ depend only on $n$, $m$, and $d$.
	We then apply \cref{cons:bnd} to the collection of $m^2$ pairs $(X_{i,j},Y_{i,j})$,
	resulting in strings $X$ and $Y$.
	The running time of this construction is $\Oh(m^2(\ell_X+\ell_Y))=\Oh(m^2\cdot \lceil \frac{n}{m}\rceil d) = \Oh(mnd)$.
	Moreover, we set $k = \frac{m^2-1}{a}(\ell_X+\ell_Y)+k_{i,j}$ (recall that $k_{i,j}$ is the same for all $i,j\in [1\dd m]$); this value satisfies $k \le \frac{m^2}{a}(\ell_X+\ell_Y)=\Oh(a^{-1}mnd)$. 
	If $\langle u,v \rangle = 0$ for some $u\in U$ and $v\in V$, then \cref{cor:lb}
	implies that $D^+_a(X_{i,j},Y_{i,j})\le k_{i,j}$ for $i,j\in [1\dd m]$ such that $u\in U_i$ and $v \in V_j$.
	Consequently, $D_a(X,Y)\le k$ by \cref{lem:ulb}.
	Symmetrically, if $D_a(X,Y)\le k$, then, by \cref{lem:ulb}, $D^+_a(X_{i,j},Y_{i,j})\le k_{i,j}$ holds for some $i,j\in [1\dd m]$ and, by \cref{cor:lb}, $\langle u,v\rangle = 0$ for some $u\in U_i\sub U$ and $v\in V_j\sub V$.
	The threshold $k$ for $D_a(X,Y)$ translates to a threshold $K = k+|X|-|Y|=k+\Oh(\ell)
	=\Oh(\max(a^{-1}mnd,nd/m))$ for $\EDA$.
\end{proof}

\thmlb*
\begin{proof}
	For each $n\in \Zp$, let $t_n = \min(n,a_nk_n)k_n$.
	Note that $\Omega(n)$ time is already needed to check whether $X=Y$.
	Thus, the claim holds trivially if there are infinitely many pairs $(a_n,k_n)$ with $t_n^{1-\epsilon} < n$.
	Consequently, we may assume without loss of generality that $t_n^{1-\epsilon }\ge n$ for each $n$ (any finite prefix of the sequence $(a_n,k_n)$ is irrelevant).
	
	Let us choose a constant $c'$ so that $(a_n,k_n)$ can be constructed in $\Oh(n^{c'})$ time,
	a constant $c$ of \cref{conj:ov} for $\epsilon/(2c')$, and a constant $C$ so that \cref{prp:lb} guarantees $|X|+|Y|\le Cnmd$ and $K \le C(\max(m/a,1/m)nd)$.
	
	Given an instance $V \sub \{0,1\}^d$ of the Orthogonal Vectors problem with $d = c\log |V|$,
	we set $n := \ceil{|V|^{1/c'}}$ and compute the values $a_n$, $k_n$, and $t_n$ in $\Oh(n^c)=\Oh(|V|)$ time. 
	Next, we set $N := \floor{\sqrt{t_n}/(Cd)}$ and cover $V$ with $v := \ceil{\frac{1}{N}|V|}$ subsets $V_1,\ldots, V_v$ of size $N$.
	For any $i,j\in [1\dd v]$, we construct an instance $(X_{i,j},Y_{i,j},K_{i,j})$ using \cref{prp:lb}
	for $V_i,V_j\sub \{0,1\}^d$, $a_n$, and $m_n := \floor{\sqrt{\min(n/k_n,a_n)}}$.
	
	Note that $|X_{i,j}|+|Y_{i,j}|\le C m_n N d \le \sqrt{\min(n/k_n,a_n)t_n} \le \sqrt{n/k_n \cdot nk_n} \le n$ and $K_{i,j}\le C \cdot \max(m_n/a_n, 1/m_n)N d \le \max(m_n/a_n, 1/m_n) \sqrt{t_n}$.
	If $a_nk_n \le n$, then $K_{i,j}\le \max(\sqrt{a_n}/a_n, 1/\sqrt{a_n})\cdot \sqrt{a_nk_n^2}=k_n$.
	If $a_nk_n \ge n$, on the other hand, then $K_{i,j}\le \max(\sqrt{n/k_n}/a_n,\sqrt{k_n/n})\cdot \sqrt{nk_n}\le \max(n/a_n,k_n)\le k_n$.
	Thus, we can test $\EDA(X_{i,j},Y_{i,j})\le K_{i,j}$ using an instance $(X_{i,j},Y_{i,j},a_n,k_n)$ of the problem in question.
	By our hypothesis, this costs $\Oh(n+t_n^{1-\epsilon})=\Oh(t_n^{1-\epsilon})$ time (including construction of \cref{prp:lb}).
	Summing up over $i,j\in [1\dd v]$, the running time is 
	$\Oh(v^2 t_n^{1-\epsilon})=\Oh(|V|^2/N^2 \cdot t_n^{1-\epsilon})=\Oh(d^2|V|^2t_n^{-\epsilon})
	=\Oh(d^2|V|^2 n^{-\epsilon}) = \Oh(d^2|V|^{2-\epsilon/c'})=\Oh(|V|^{2-\epsilon/(2c')})$.
	Overall, including the construction of $(a_n,k_n)$ in $\Oh(|V|)$ time, the total time complexity of solving the OV instance is $\Oh(|V|^{2-\epsilon/(2c')})$, contradicting \cref{conj:ov}.
	\end{proof}

\ifArXiv
\appendix

\section{Bicriteria Algorithms}\label{app:bic}

\subsection{Exact Algorithm}
\biexact*

Our algorithm runs in time $\Oh(n+k_Sk_I^2)$ and is implemented as \cref{alg:biexact},
where all uninitialized and out-of-bounds values $\De_{v_I,v_S}[s]$
are implicitly set to $-\infty$.
Its correctness follows from \cref{lem:biexact} (see below),
and the running time is proportional to the number of $\LCE$ queries, which is $\Oh(k_Sk_I^2)$, plus the $\Oh(n)$ construction time 
of an $\LCE$ data structure.

\begin{algorithm}[htb]
	\ForEach{$v_S\in [0\dd k_S]$}{%
		\ForEach{$v_I\in [0\dd k_I]$}{
		\ForEach{$s\in [-v_I\dd v_I]$}{
			$\De'_{v_I,v_S}[s] \gets \min(|X|,|Y|-s,\max(\De_{v_I-1,v_S}[s-1],\De_{v_I,v_S-1}[s]+1,\De_{v_I-1,v_S}[s+1]+1))$\;
			\lIf{$s=0$}{$\De'_{v_I,v_S}[s] \gets \max(\De'_{v_I,v_S}[s],0)$}
			$\De_{v_I,v_S}[s] \gets \De'_{v_I,v_S}[s] + \LCE(\De'_{v_I,v_S}[s],\De'_{v_I,v_S}[s]+s)$\;
		}
		\lIf{$\De_{v_I,v_S}[|Y|-|X|]=|X|$}{\Return{YES}}
		}
	}
	\Return{NO}\;
	
	\caption{Exact bicriteria algorithm}\label{alg:biexact}
	
\end{algorithm}

\begin{lemma}\label{lem:biexact}
	For every $v_I\in [0\dd k_I]$, $v_S\in [0\dd k_S]$, $x\in [0\dd |X|]$, and $y\in [0\dd |Y|]$, we have $\De_{v_I,v_S}[y-x]\ge x$
	if and only if $X[0\dd x)$ and $Y[0\dd y)$ admit an $(v_I,v_S)$-alignment.
\end{lemma}
\begin{proof}
	Let us first prove, by induction on $v_I+v_S$, that $\De_{v_I,v_S}[y-x]\ge x$ if  $X[0\dd x)$ and $Y[0\dd y)$ admit a $(v_I,v_S)$-alignment. 
	Let us fix such an alignment and consider the longest suffixes $X[x'\dd x)=Y[y'\dd y)$ aligned without any edit.
	We shall prove that $\De'_{v_I,v_S}[y-x]=\De'_{v_I,v_S}[y'-x']\ge x'$ by considering four cases:
	\begin{itemize}
		\item If $x'=y'=0$, then $\De'_{v_I,v_S}[y'-x']=\De'_{v_I,v_S}[0]\ge 0$.
		\item If $X[x'-1]$ is deleted, then $X[0\dd x'-1)$ and $Y[0\dd y')$ admit a $(v_I-1, v_S)$-alignment.
		By inductive assumption, $\De'_{v_I,v_S}[y'-x']\ge \De_{v_I-1,v_S}[y'-x'+1]+1\ge x'$.
		\item If $Y[y'-1]$ is inserted,  then $X[0\dd x')$ and $Y[0\dd y'-1)$ admit a $(v_I-1, v_S)$-alignment.
		By inductive assumption,  $\De'_{v_I,v_S}[y'-x']\ge \De_{v_I-1,v_S}[y'-x'-1]\ge x'$.
		\item Otherwise, $X[x'-1]$ is substituted for $Y[y'-1]$; then, $X[0\dd x'-1)$ and $Y[0\dd y'-1)$ admit a $(v_I,v_S-1)$-alignment.
		By inductive assumption,  $\De'_{v_I,v_S}[y'-x']\ge \De_{v_I,v_S-1}[y'-x']+1\ge x'$.
	\end{itemize}
	Now, $\De_{v_I,v_S}[y-x]\ge x$ follows from $\De'_{v_I,v_S}[y'-x']\ge x'$ due to $\LCE(x',y')\ge x-x'$.
	
	The converse implication is also proved by induction on $v_I+v_S$. We consider four cases 
	depending on $x':=\De'_{v_I,v_S}[y-x]$ and $y' := x'+(y-x)$:
	\begin{itemize}
		\item If $x'\le \De_{v_I-1,v_S}[y'-x'+1]+1$, then, by the inductive assumption, $X[0\dd x'-1)$ and $Y[0\dd y')$ admit a $(v_I-1, v_S)$-alignment,
		which can be extended to a $(v_I, v_S)$-alignment of $X[0\dd x')$ and $Y[0\dd y')$ by deleting $X[x'-1]$.
		\item If $x'\le \De_{v_I-1,v_S}[y'-x'-1]$, then, by the inductive assumption, $X[0\dd x')$ and $Y[0\dd y'-1)$ admit a $(v_I-1, v_S)$-alignment,
		which can be extended to a $(v_I, v_S)$-alignment of $X[0\dd x')$ and $Y[0\dd y')$ by inserting $Y[y'-1]$.
		\item If $x'\le \De_{v_I,v_S-1}[y'-x']+1$, then, by the inductive assumption, $X[0\dd x'-1)$ and $Y[0\dd y'-1)$ admit a $(v_I, v_S-1)$-alignment,
		which can be extended to a $(v_I,v_S)$-alignment of $X[0\dd x')$ and $Y[0\dd y')$ by substituting $X[x'-1]$ for $Y[y'-1]$.
		\item In the remaining case, we have $x'=y'=0$. Trivially, $X[0\dd x')$ and $Y[0\dd y')$ admit a $(0,0)$-alignment, which is also a $(v_I,v_S)$-alignment.
	\end{itemize}
	In all cases, the $(v_I,v_S)$-alignment of  $X[0\dd x')$ and $Y[0\dd y')$ yields a $(v_I,v_S)$-alignment of $X[0\dd x)$ and $Y[0\dd y)$
	because $\LCE(x',y')\ge x-x'$.
\end{proof}

\subsection{Approximation algorithm}
Given strings $X,Y$, integer thresholds $k_I,k_S$, and an accuracy parameter $\epsilon\in(0,1)$,
our next algorithm reports YES and $X$ and $Y$ admit a $(k_I,k_S)$-alignment,
NO if $X$ and $Y$ do not admit a $(k_I, (1+\epsilon)k_S)$-alignment,
and an arbitrary answer otherwise.
The algorithm mimics the behavior of~\cref{alg:biexact} with a coarser granularity of the substitution costs.
We assume that out-of-bounds and uninitialized values $\Da_{v_I,v_S}[s]$ are implicitly set to $-\infty$.

\begin{algorithm}[htb]
	$d := \ceil{\frac{\epsilon k_S}{2+k_I}}$\;
	\ForEach{$v_S\in \{0,d,2d,\ldots,d\cdot \ceil{k_S/d}\}$}{%
		\ForEach{$v_I\in [0\dd k_I]$}{
		\ForEach{$s\in [-v_I\dd v_I]$}{
			$\Da'_{v_I,v_S}[s] \gets \min(|X|,|Y|-s,\max(\Da_{v_I-1,v_S}[s-1],\Da_{v_I,v_S-d}[s],\Da_{v_I-1,v_S}[s+1]+1))$\;
			\lIf{$s=0$}{$\Da'_{v_I,v_S}[s] \gets \max(\Da'_{v_I,v_S}[s],0)$}
			$\Da_{v_I,v_S}[s] \gets \Da'_{v_I,v_S}[s] + \LCE_{d,\epsilon}(\Da'_{v_I,v_S}[s],\Da'_{v_I,v_S}[s]+s)$\;
		}
		\lIf{$\Da_{v_I,v_S}[|Y|-|X|]=|X|$}{\Return{YES}}
		}
	}
	\Return{NO}\;
	
	\caption{Approximate bicriteria algorithm}\label{alg:biapx}
	
\end{algorithm}

The following lemma justifies the correctness of \cref{alg:biapx}.
\begin{lemma}\label{lem:biapx}
Let $d=\ceil{\frac{\epsilon k_S}{2+k_I}}$.
For every $v_S\in \{0,d,\ldots,d \ceil{k_S/d}\}$, $v_I\in [0\dd k_I]$, $x\in [0\dd |X|]$, and $y\in [0\dd |Y|]$:
\begin{itemize}
	\item If $X[0\dd x)$ and $Y[0\dd y)$ admit a $(v_I,v_S)$-alignment, then $\Da_{v_I,v_S}[y-x]\ge x$.
	\item If $\Da_{v_I,v_S}[y-x]\ge x$, then $X[0\dd x)$ and $Y[0\dd y)$ admit a $(v_I, (1+\epsilon)(d+v_Id+v_S))$-alignment.
\end{itemize}
\end{lemma}
\begin{proof}
	We start by proving the first implication inductively on $v_I+v_S$.
	Let us fix a $(v_I,v_S)$-alignment
	of $X[0\dd x)$ and $Y[0\dd y)$
	and consider the longest suffixes $X[x'\dd x)$ and $Y[y'\dd y)$ aligned with at most $d$ substitutions and no indels. 
	We shall prove that $\Da'_{v_I,v_S}[y-x]=\Da'_{v_I,v_S}[y'-x']\ge x'$ by considering four cases:
	\begin{itemize}
		\item If $x'=y'=0$, then $\Da'_{v_I,v_S}[y'-x'] = \Da'_{v_I,v_S}[0] \ge 0$.
		\item If $X[x'-1]$ is deleted, then $X[0\dd x'-1)$ and $Y[0\dd y')$ admit a $(v_I-1, v_S)$-alignment.
		By the inductive assumption, $\Da'_{v_I,v_S}[y'-x']\ge \Da_{v_I-1,v_S}[y'-x'+1]+1\ge x'$.
		\item If $Y[y'-1]$ is inserted, then $X[0\dd x')$ and $Y[0\dd y'-1)$ admit a $(v_I-1, v_S)$-alignment.
		By the inductive assumption, $\Da'_{v_I,v_S}[y'-x']\ge \Da_{v_I-1,v_S}[y'-x'-1]\ge x'$.
		\item Otherwise, there are exactly $d$ substitutions between $X[x'\dd x)$ and $Y[y'\dd y)$.
		Thus, $X[0\dd x')$ and $Y[0\dd y')$ admit a $(v_I,v_S-d)$-alignment and, by the inductive assumption, $\Da'_{v_I,v_S}[y'-x'] \ge \Da_{v_I,v_S-d}[y'-x']\ge x'$.
	\end{itemize}
	Now, $\Da_{v_I,v_S}[y-x]\ge x$ follows from $\Da'_{v_I,v_S}[y'-x'] \ge x'$ due to $\LCE_{d}(x',y')\ge x-x'$.

	The second implication is also proved by induction on $v_I+v_S$.
	We consider four cases depending on $x':=\Da'_{v_I,v_S}[y-x]$ and $y' := x'+(y-x)$:
	\begin{itemize}
		\item If $x'\le \Da_{v_I-1,v_S}[y'-x'+1]+1$, then, by the inductive assumption, $X[0\dd x'-1)$ and $Y[0\dd y')$ admit a $(v_I-1, (1+\epsilon)(v_Id+v_S))$-alignment,
		which can be extended to a $(v_I, (1+\epsilon)(v_Id+v_S))$-alignment of $X[0\dd x')$ and $Y[0\dd y')$ by deleting $X[x'-1]$.
		\item If $x'\le \Da_{v_I-1,v_S}[y'-x'-1]$, then, by the inductive assumption, $X[0\dd x')$ and $Y[0\dd y'-1)$ admit a $(v_I-1, (1+\epsilon)(v_Id+v_S))$-alignment,
		which can be extended to a $(v_I, (1+\epsilon)(v_Id+v_S))$-alignment of $X[0\dd x')$ and $Y[0\dd y')$ by inserting $Y[y'-1]$.
		\item If $x'\le \Da_{v_I,v_S-1}[y'-x']$, then, by the inductive assumption, $X[0\dd x')$ and $Y[0\dd y')$ admit a $(v_I, (1+\epsilon)(v_Id+v_S))$-alignment.
		\item In the remaining case, we have $x'=y'=0$. Trivially, $X[0\dd x')$ and $Y[0\dd y')$ admit a $(0,0)$-alignment, which is also a $(v_I, (1+\epsilon)(v_Id+v_S))$-alignment.
	\end{itemize}
	In all cases, the $(v_I, (1+\epsilon)(v_Id+v_S))$-alignment of $X[0\dd x')$ and $Y[0\dd y')$  yields a $(v_I, \allowbreak (1+\epsilon)(d+v_Id+v_S))$-alignment of $X[0\dd x)$ and $Y[0\dd y)$ because $\LCE_{(1+\epsilon)d}(x',y')\ge x-x'$.
\end{proof}

\begin{corollary}\label{cor:biapx}
\cref{alg:biapx} returns YES if $X$ and $Y$ admit a $(k_I,k_S)$-alignment,
and NO if $X$ and $Y$ do not admit a $(k_I,(1+\epsilon)(2+k_I+(1+\epsilon)k_S))$-alignment.
Moreover, it can be implemented in $\tOh(\frac{n k_I}{\epsilon^3 k_S}+k_I^3 k_S)$ time.
\end{corollary}
\begin{proof}
First, suppose that $X$ and $Y$ admit a $(k_I,k_S)$-alignment.
Then, we apply \cref{lem:biapx} for $v_S = d \ceil{k_S/d}$, $v_I=k_I$, $x=|X|$, and $y=|Y|$.
We conclude that $\Da_{v_I,v_S}[|X|-|Y|]\ge |X|$, which means that the algorithm returns YES.
On the other hand, if the algorithm returns YES, then $\Da_{v_I,v_S}[|X|-|Y|]\ge |X|$
holds for some $v_I\in [0\dd k_I]$ and $v_S\in \{0,d,\ldots,d \ceil{k_S/d}\}$.
In this case, \cref{lem:biapx} implies that $X$ and $Y$ admit a $(v_I, (1+\epsilon)(d+v_Id+v_S))$-alignment,
which is also a $(k_I,(1+\epsilon)(k_S+k_I+2+\epsilon))$-alignment because $v_I\le k_I$
and $(1+\epsilon)(d+v_Id+v_S) \le (1+\epsilon)(2d+k_I+k_S)(1+\epsilon) \le (1+\epsilon)((2+k_I)d + k_S)
\le (1+\epsilon)(2+k_I+(1+\epsilon)k_S)$.
As for the running time, we observe that the number of $\LCE_{d,\epsilon}(x,y)$ queries is $\Oh(k_I^2 \cdot k_S/d)$ and that each of them satisfies $|x-y|\le k_I$. By \cref{prp:LCE}, the total running time is therefore $\tOh(\frac{n k_I}{\epsilon^3 k_S}+k_I^3 \cdot k_S)$.
\end{proof}

\biapx*
\begin{proof}
	If $\frac{\epsilon k_S}{5} < 2+k_I$, we use the algorithm of \cref{prp:bi}, which costs $\Oh(n+k_S k_I^2)= \Oh(\frac{n k_I}{\epsilon k_S}+k_S k_I^2)
	=\tOh(\frac{nk_I}{\epsilon^3 k_S}+k_S k_I^3)$ time.
	Otherwise, we apply \cref{cor:biapx} with the accuracy parameter $\epsilon':=\frac{\epsilon}{5}$.
	This guarantees that $(1+\epsilon')(2+k_I+(1+\epsilon')k_S)\le (1+\epsilon')(1+2\epsilon')k_S \le (1+5\epsilon')k_S
	= (1+\epsilon)k_S$. Thus, \cref{alg:biapx} meets the requirements of \cref{prob:bi} with $\alpha=1$ and $\beta=1+\epsilon$.
\end{proof}
\fi

\bibliographystyle{plainurl}
\bibliography{references}

\end{document}